\documentclass[10pt, english]{article}

\usepackage[utf8]{inputenc} 
\usepackage[T1]{fontenc}    
\usepackage[colorlinks,citecolor=blue,urlcolor=blue,linkcolor=blue,linktocpage=true]{hyperref}       
\usepackage{url}            
\usepackage{booktabs}       
\usepackage{amsfonts}       
\usepackage{nicefrac}       
\usepackage{microtype}      
\usepackage{xcolor}         

\usepackage{amsmath}
\usepackage{amsmath}
\usepackage{amssymb}
\usepackage{mathtools}
\usepackage{amsthm}
\usepackage{caption}

\usepackage[toc,page,header]{appendix}
\usepackage{minitoc}

\usepackage{subcaption}

\usepackage{microtype}
\usepackage{graphicx}
\usepackage{booktabs}

\usepackage{algorithmic}
\usepackage{algorithm}
\usepackage{xfrac}

\usepackage[round]{natbib}
\usepackage{chngcntr}
\usepackage{apptools}

\DeclareMathOperator*{\argmin}{arg\,min}

\theoremstyle{plain}
\newtheorem{theorem}{Theorem}[section]

\newtheorem{lemma}[theorem]{Lemma}
\newtheorem{corollary}[theorem]{Corollary}

\theoremstyle{definition}
\newtheorem{definition}[theorem]{Definition}

\theoremstyle{remark}
\newtheorem{remark}[theorem]{Remark}

\numberwithin{equation}{section}
\newcommand{\ee}{{\rm e}\hspace{1pt}}

\newcommand{\dd}{\hspace{1pt}{\rm d}\hspace{0.5pt}}
\newcommand{\abs}[1]{\left| #1 \right|}

\newcommand{\wt}{\widetilde}
\newcommand{\wh}{\widehat}
\DeclarePairedDelimiter{\ceil}{\lceil}{\rceil}

\title{Privacy Profiles for Private Selection}

\author{Antti Koskela$^1$, Rachel Redberg$^2$, Yu-Xiang Wang$^2$ \vspace{5mm} \\
$^1$ Nokia Bell Labs \\
$^2$ UC Santa Barbara }

    \date{}

\begin{document}
	

\maketitle

\begin{abstract}

Private selection mechanisms (e.g., Report Noisy Max, Sparse Vector) are fundamental primitives of  differentially private (DP) data analysis with wide applications to private query release, voting, and hyperparameter tuning. Recent work \citep{liu2019private,papernot2022hyperparameter} has made significant progress in both generalizing private selection mechanisms and tightening their privacy analysis using modern numerical privacy accounting tools, e.g., R\'enyi differential privacy (RDP).  But R\'enyi DP is known to be lossy when $(\epsilon,\delta)$-DP is ultimately needed, and there is a trend to close the gap by directly handling privacy profiles, i.e., $\delta$ as a function of $\epsilon$ or its equivalent dual form known as $f$-DPs. In this paper, we work out an easy-to-use recipe that bounds the privacy profiles of  ReportNoisyMax and PrivateTuning using the privacy profiles of the base algorithms they corral.  Numerically, our approach improves over the RDP-based accounting in all regimes of interest and leads to substantial benefits in end-to-end private learning experiments. Our analysis also suggests new distributions, e.g., binomial distribution for randomizing the number of rounds that leads to more substantial improvements in certain regimes.

\end{abstract}

\section{Introduction} \label{sec:introduction}

Differential privacy (DP) bounds the privacy loss incurred when an algorithm is run on a dataset. While the analysis of this bound is often quite nuanced, the rough tally is that whenever the algorithm accesses the data, it incurs a privacy cost.

By this rule of thumb, data privacy for modern machine learning (ML) applications is in trouble. Most modern ML algorithms are notoriously finicky and require extensive hyperparameter tuning in order to achieve good performance. In the context of data privacy, this means that the evaluation of every additional hyperparameter candidate could cost something.

Using composition theorems, running an $(\epsilon, \delta)$-DP algorithm $k$ times incurs a cost of $(k \epsilon, k \delta)$-DP or $(\tilde{O}(\sqrt{k}\epsilon),O(\delta))$-DP. 
This naïve analysis assumes a meta-algorithm that outputs a $k$-tuple of results over all $k$ runs of the $(\epsilon, \delta)$-DP base algorithm. But for applications such as hyperparameter tuning, the meta-algorithm need only output the ``best’’ output over all $k$ runs. Thus, a tighter analysis may often be available for this class of \emph{private selection} mechanisms which choose from a set of candidates the item which approximately maximizes a given quality score.

The aim of our work is to improve the privacy analysis of private hyperparameter tuning algorithms for machine learning algorithms. Previous studies ~\citep{chaudhuri2013stability} have already addressed the DP cost of hyperparameter tuning. Improving upon this, \citet{liu2019private} proposed black-box tuning methods for which they showed that private selection from $\epsilon$-DP base mechanisms can be $3\epsilon$-DP, and also gave approximate $(\epsilon,\delta)$-DP results.~\citet{mohapatra2022role} gave adaptive tuning methods and showed that for a reasonable number of privately chosen candidates, naïve accounting via Rényi differential privacy (RDP) often yields tighter DP bounds than private selection using the methods by~\citet{liu2019private}. We base our approach on the results of~\citet{papernot2022hyperparameter} which build on \citet{liu2019private}'s work.~\citet{papernot2022hyperparameter} provided logarithmically scaling RDP bounds for black-box tuning algorithms in relation to the number of model evaluations. 

What is missing from the private selection and private hyperparameter tuning line of work, are bounds that would directly use the privacy profiles of the base mechanisms as building blocks and would also be capable of taking advantage of numerical accounting methods ~\citep[e.g., the recent work of][]{koskela2021tight,zhu2021optimal,gopi2021}.
In this work we address this shortfall.
Our proposed bounds utilize only point-wise information about the $(\epsilon,\delta)$-profile of the base mechanism, similarly to bounds given by~\citet{liu2019private}. 

As one application of our results we consider Differentially Private Stochastic Gradient Descent (DP-SGD)~\citep{Abadi2016}, a popular technique for training machine learning models with DP. DP-SGD introduces extra hyperparameters such as the noise level $\sigma$ and clipping constant $C$; factors like the subsampling ratio $q$ and the training duration also affect both the privacy and accuracy of the methods. As demonstrated by~\citet{papernot2022hyperparameter}, tweaking DP-SGD's hyperparameters often relies on using sensitive data which require privacy protection. The risk of data leakage from hyperparameters is arguably smaller than from model parameters, but developing methods with low privacy costs has proven challenging. The leading algorithms proposed by~\citet{papernot2022hyperparameter} still incur a significant privacy cost overhead. Our aim is to further lower this overhead.

As another application of our analysis we consider Generalized Propose-Test-Release~\citep{redberg2023generalized}, which broadens the reach of the Propose-Test-Release (PTR) framework by allowing it to handle queries with unbounded sensitivity, making it applicable to problems such as linear regression.
The method proposed by~\citet{redberg2023generalized} gives point-wise $(\epsilon,\delta)$-privacy guarantees for generalized PTR. Thus, in order to account for the privacy cost of tuning the hyperparameters of the underlying queries, we can directly use our analysis for which point-wise  $(\epsilon,\delta)$-guarantees are sufficient.
We show that our bounds are considerably tighter than those of~\citet{liu2019private}, the only previous applicable results. We also empirically illustrate that compared to well-established non-adaptive methods, our bounds considerably improve the privacy-utility trade-off for linear regression problems.

\section{Preliminaries} \label{sec:preliminaries}


We first give the basic definitions.
An input dataset containing $n$ data points is denoted as $X = \{x_1,\ldots,x_n\}$. 
Denote the set of all possible datasets by $\mathcal{X}$.
We say $X$ and $X'$ are neighbors if we get one by adding or removing one data element to or from the other, or by replacing one data element in the other (denoted $X \sim X'$). 
Consider a randomized mechanism $\mathcal{M} \, : \, \mathcal{X} \rightarrow \mathcal{O}$, where  $\mathcal{O}$ denotes the output space.
The $(\epsilon,\delta)$-definition of DP can be given as follows~\citep{dwork06differential}.
\begin{definition} \label{def:indistinguishability}
	Let $\epsilon > 0$ and $\delta \in [0,1]$.
	We say that a mechanism $\mathcal{M}$ is $(\epsilon,\delta)$-DP, if for all neighboring datasets $X$ and $X'$ and
	 for every measurable set $E \subset \mathcal{O}$ we have:
	\begin{equation*}
		\begin{aligned}
			\mathrm{Pr}( \mathcal{M}(X) \in E ) \leq \ee^\epsilon \mathrm{Pr} (\mathcal{M}(X') \in E ) + \delta.
		\end{aligned}
	\end{equation*}
\end{definition}
We state many of our results for general $f$-divergences.
For a convex function $f  :  [0,\infty) \rightarrow \mathbb{R}$, we define the $f$-divergence between distributions $P$ and $Q$ taking values in $\mathcal{Y}$ as
$$
H_f (P || Q) = \int_{\mathcal{Y}}  f \left( \frac{ P(y)   }{ Q(y) }\right) Q(y) \, \dd y.
$$
Notice that we do not require the normalization $f(1)=0$ often used in the so-called Czsis\'ar divergences~\citep{liese2006divergences}.
Especially, our aim is to find tight bounds for the hockey stick divergence, i.e., when
$f(z) = [z - \ee^{\epsilon}]_+$
for some $\epsilon \in \mathbb{R}$. This is due to the fact that
tight $(\epsilon,\delta)$-bounds can be obtained using the hockey-stick-divergence:
\begin{lemma}[{\citealt[Theorem~1]{balle2018subsampling}}] 
A mechanism $\mathcal{M}$ satisfies $(\epsilon,\delta)$-DP \emph{if and only if},  
$$
\max_{X \sim X'} H_{f}(\mathcal{M}(X)||\mathcal{M}(X'))\leq \delta
$$
for $f(z) = [z - \ee^{\epsilon}]_+$.
\end{lemma}
We denote the hockey stick divergence determined by $\epsilon \in \mathbb{R}$ by $H_{\ee^{\epsilon}}$ throughout the paper, and will refer to
\begin{equation*} 
\delta_{\mathcal{M}}(\epsilon) := \max_{X \sim X'} H_{e^\epsilon}(\mathcal{M}(X)||\mathcal{M}(X'))
\end{equation*}
as the \emph{privacy profile} of mechanism $\mathcal{M}$. 

We will also use the R\'enyi differential privacy (RDP)~\citep{mironov2017} which is defined as follows. R\'enyi divergence of order $\alpha > 1$ between two distributions $P$ and $Q$ is defined as 
\begin{equation} \label{eq:rdp}
D_\alpha(P || Q) = \frac{1}{\alpha - 1} \log \int \left( \frac{P(t)}{Q(t)} \right)^\alpha Q(t) \, \dd t
\end{equation}
and we say that a mechanism $\mathcal{M}$ is $(\alpha,\epsilon)$-RDP, if for all neighboring datasets $X$ and $X'$, the output distributions of
$\mathcal{M}(X)$ and $\mathcal{M}(X')$ have R\'enyi divergence of order $\alpha>1$ at most $\epsilon$, i.e., if 
$$ 
\max_{X \sim X'} D_\alpha\big( \mathcal{M}(X) || \mathcal{M}(X') \big)  \leq \epsilon.
$$
To convert from RDP to $(\epsilon,\delta)$-DP, we use the formula given in Appendix~\eqref{eq:conversion_canonne}.

Notice that the R\'enyi divergence of order $\alpha>1$ is a scaled logarithm of an $f$-divergence determined by $f(z) = z^\alpha$.
Existing RDP analyses for RNM and private selection use the fact $f(u,v) = (\tfrac{u}{v})^\alpha v$ is a jointly convex function for $u$ and $v$. We formulate many of our results for general $f$-divergences, and the following technical result will then be central~\citep[see e.g., Ch.\;3,][]{boyd2004convex}.
\begin{lemma} \label{lem:f_convex}
If $f \, : \, \mathbb{R}_+ \rightarrow \mathbb{R}_+ $ is a convex function, then the function 
$f\left( \frac{x}{y} \right) y$ is jointly convex for $x \geq 0$ and $y>0$.
\end{lemma}
For the analysis of private selection algorithms, the number of times $K$ that the base algorithm is evaluated is a random variable. To analyze different alternatives for choosing $K$, we will need the concept of probability generating functions.
\begin{definition}
    Let $K$ be a random variable taking values in $\mathbb{N} \cup \{0\}.$ The probability generating function (PGF) of $K$, $\varphi \, : \, [0,1] \rightarrow [0,1]$ is defined as 
    $$
    \varphi(z) = \sum\limits_{k=0}^\infty \mathbb{P}(K=k) \cdot z^k.
    $$
\end{definition}
Our main result Thm.~\ref{thm:main_thm} is stated for a general PGF $\varphi$ and we use it obtain method-specific bounds.
Throughout the paper, we will denote $m = \mathbb{E}[K]$.

\section{Report Noisy Max for Additive Noise Mechanisms} \label{sec:max_sel}

As a first application of the hockey-stick divergence-based analysis, we consider the Report Noisy Max (RNM) of 1-dimensional additive noise mechanisms. This will serve as a segue into the more involved analysis of private selection, where we also obtain bigger gains compared to the previous results. Our analysis here is based on the RDP analysis by~\citet{zhu2022adaptive}. We mention that recent applications of RNM include private in-context learning of LLMs~\citep{wu2024privacy,tang2024privacy}.

Let $X=\{x_1,x_2,\ldots,x_N\}$, where $x_i \in \mathcal{X}$ for all $i \in [N]$, be a data set and consider the mechanism
\begin{equation} \label{eq:M_argmax}
    \mathcal{M}(X) = \arg \max \{ \mathcal{M}_1(X), \ldots,\mathcal{M}_m(X) \},
\end{equation}
where for every $i \in [m]$,
$$
\mathcal{M}_i(X) = f(X) + Z_i,
$$
for some function $f_i \; : \; \mathcal{X} \rightarrow \mathbb{R}$ such that 
$$
\max_{X \sim X'} \abs{f_i(X) - f_i(X')} \leq 1
$$
and where the noises $Z_i$, $i \in [m]$, are i.i.d.
We have the following existing RDP bound. 
\begin{theorem}[{\citealt[Theorem~8]{zhu2022adaptive}}] \label{lem:zhu_lemma5}
Let $\alpha>1$. The mechanism $ \mathcal{M}(X)$ is $(\alpha,\epsilon')$-RDP for
\begin{equation} \label{eq:zhuwang22RDP}
    \epsilon' = \epsilon + \frac{\log m}{\alpha - 1},
\end{equation}
where $\epsilon$ denotes the RDP guarantee of order $\alpha$ for an additive noise mechanism with noise $Z$ and sensitivity 2.
\end{theorem}

\begin{theorem} \label{thm:max_sel}
Let $X \sim X'$ and $\epsilon \in \mathbb{R}$. We have:
\begin{equation} \label{eq:max_sel_bound}
    H_{\ee^\epsilon} \big(\mathcal{M}(X) || \mathcal{M}(X') \big) \leq m \cdot \delta(\epsilon),
\end{equation}
where $\delta(\epsilon)$ is the privacy profile of the additive noise mechanism with sensitivity 2. If we assume monotonicity, i.e., if $f_i(X) \geq f_i(X')$ (or vice versa) for all $i \in [m]$, then the sensitivity of 2 can be replaced with a sensitivity of 1. 
\end{theorem}



\begin{theorem}\label{lem:composition_rnm}
For an adaptive composition of $k$ mechanisms of the form~\eqref{eq:M_argmax}, we get the 
privacy profile upper bound $m^k \cdot \delta(\epsilon)$, where $\delta(\epsilon)$ is the privacy profile of an $m$-wise composition of an additive noise mechanism with noise $Z$ and sensitivity $2$. 
\end{theorem}

The RDP analysis of the private selection provided by~\citet{papernot2022hyperparameter} shows that the RDP guarantees grow essentially as $\log m$, where $m$ is the expected number of candidates for the private selection algorithm. In case $Z$ is normally distributed, we directly see from our analysis that for a fixed $\delta$ the
$\epsilon$-values of the private selection algorithm grow as $\tfrac{1}{\sigma} \log^{\frac{1}{2}} \frac{m}{\delta}$. This result
is obtained for the RNM using a simple tail bound of the Gaussian.

\begin{corollary} \label{lem:convert2_eps_delta}
Consider the mechanism $\mathcal{M}$ defined in Eq.~\ref{eq:M_argmax} and suppose $Z$ is normally distributed with variance $\sigma^2$. 
Let $\delta>0$. Then $\mathcal{M}$ is $(\epsilon,\delta)$-DP for 
$$
\epsilon = \frac{2}{\sigma^2} + \frac{2}{\sigma} \sqrt{ 2 \log \frac{m}{\delta} }
$$
\end{corollary}
As Figure~\ref{fig:comparison_argmax2} shows, with the bound given in Thm.~\ref{thm:max_sel} we observe differences one usually observes between the accurate hockey stick and R\'enyi divergence bounds~\citep[for comparisons, see, e.g.][]{canonne2020discrete}.
When considering the private selection problem where the number of candidates is randomized, we obtain larger differences. Also, the $\epsilon$-growth order $\mathcal{O}(\tfrac{1}{\sigma} \log^{\frac{1}{2}} \frac{m}{\delta})$ is retained. All of this is discussed in the next two sections.

\begin{figure} [h!]
     \centering
        \includegraphics[width=.45\textwidth]{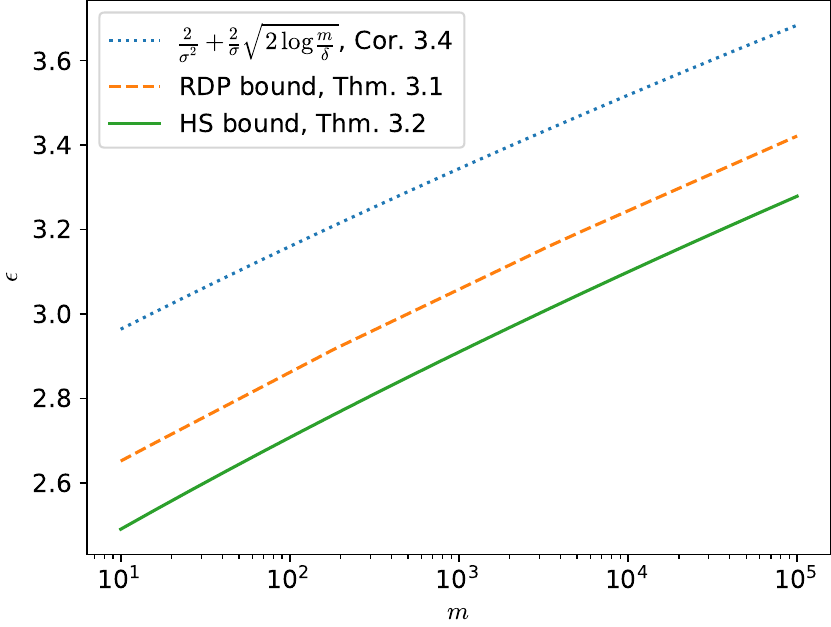}
        \caption{Comparison of the $(\epsilon,\delta)$-bounds for the RNM mechanism~\eqref{eq:M_argmax} when the base mechanisms $\mathcal{M}_i$, $i \in [m]$, are 1-d Gaussian mechanisms with sensitivity 1 and noise scale $\sigma=4.0$, and when $\delta=10^{-6}$. Also plotted is the bound of Corollary~\ref{lem:convert2_eps_delta}. }
 	\label{fig:comparison_argmax2}
\end{figure}

\section{General Bound for Private Selection} \label{sec:priv_sel}

We use the notation and setting of~\citet{papernot2022hyperparameter}. This means that the tuning algorithm $A$ outputs both the argument of the maximizer (the best hyperparameters or the index) and the output of the base mechanism (e.g., the model trained with the best hyperparameters).

Let $\mathcal{Y}$ denote the finite and ordered output space (the quality score), $Q(y)$ the density function of the quality score of the base mechanism taking values in $\mathcal{Y}$, $K$ the random variable for the number of times the base mechanism is run and $A(y)$ the density function of the tuning algorithm that outputs the best one of the $K$ alternatives.
Let $A$ and $A'$ denote the output distributions of the tuning algorithm evaluated on neighboring datasets $X$ and $X'$, respectively.
Then, the $f$-divergence between $A$ and $A'$ can be bounded using the following result.
We use the proof technique of~\citet{papernot2022hyperparameter} and similarly invoke the argument that our proof for finite $\mathcal{Y}$ can be extended to the general case.

\begin{theorem} \label{thm:main_thm}

Let $X \sim X'$ and let $A$ and $A'$ be the density functions of the hyperparameter tuning algorithm as defined above, evaluated on $X$ and $X'$, respectively. Let $Q$ and $Q'$ be the density functions of the quality score of the base mechanism, evaluated on  $X$ and $X'$, respectively. Let $K$ be random variable for the times the base mechanism is run and $\varphi(z)$ the PGF of $K$.
Let $f \, : \,  [0,\infty) \rightarrow \mathbb{R}$ be a convex function.
Then, 
\begin{equation*} 
\begin{aligned}
H_f (A || A') 
 \leq \sum_{y \in \mathcal{Y}} 
f \left( \frac{ Q(y)  \varphi'(q_y)   }{ Q'(y) \varphi'(q'_y) }\right) \cdot Q'(y) \varphi'(q'_y),
\end{aligned}
\end{equation*}
where for each $y \in \mathcal{Y}$, $q_y$ and $q'_y$ are obtained by applying the same 
$y$-dependent post-processing function to
$Q$ and $Q'$, respectively.
\end{theorem}

Looking at the bound given by Thm.~\ref{thm:main_thm}, we can decompose the right-hand side 
in case $f$ corresponds to the R\'enyi divergence and obtain~\citep[Lemma 7,][]{papernot2022hyperparameter} as a corollary.

\begin{remark}

Let $f(z) = z^\alpha$ for some $\alpha > 1$. Then, by Thm.~\ref{thm:main_thm},
\begin{equation} \label{eq:decompose}
\begin{aligned}
& \ee^{(\alpha-1) D_\alpha (A || A')} = H_f(A || A')  \\
&\leq \sum_{y \in \mathcal{Y}} \left( \frac{ Q(y)  }{ Q'(y) }\right)^\alpha  Q'(y) \cdot
\left( \frac{  \varphi'(q_y)   }{ \varphi'(q'_y) }\right)^\alpha  \varphi'(q'_y)  \\
& \leq \left[ \sum_{y \in \mathcal{Y}} \left( \frac{ Q(y)  }{ Q'(y) }\right)^\alpha  Q'(y) \right] 
\left( \frac{  \varphi'(q)   }{ \varphi'(q') }\right)^\alpha  \varphi'(q'), 
\end{aligned}
\end{equation}
for some $q$ and $q'$ that are obtained by applying the same post-processing function to $Q$ and $Q'$, respectively.
Taking the logarithm and dividing by $\alpha - 1$, we obtain~\citep[Lemma 7,][]{papernot2022hyperparameter}. 

\end{remark}

In case of RDP analysis, the bounds for the randomized private selection algorithms (see Theorems~\ref{thm:rdp_negbin} and~\ref{thm:rdp_poisson}) allow optimizing the bound~\eqref{eq:decompose} with respect to the privacy profile of $Q$. For example: as $q$ is a result of post-processing of $Q$, we have that for all $\epsilon \geq 0$
$$
q \leq \ee^\epsilon q' + \delta(\epsilon),
$$
where $\delta(\epsilon)$ gives the privacy profile of $Q$, and we can carry out an optimization of the bound~\eqref{eq:decompose}
individually for each RDP order $\alpha$ w.r.t. $\epsilon$. In case the function $f$ in Thm.~\ref{thm:main_thm} corresponds to the hockey stick divergence, the best we can have is a uniform bound for the ratio $\frac{ \varphi'(q_y)   }{  \varphi'(q'_y) }$ which uses the bound $q_y \leq \ee^\epsilon q_y' + \delta(\epsilon)$ with the same value of $\epsilon$ for each $y \in \mathcal{Y}$. As a result there is one degree of freedom less to optimize in the upper bounds we obtain using the hockey stick divergence.
Nevertheless,  the analysis with the hockey stick divergence becomes much simpler and the resulting $(\epsilon,\delta)$-DP bounds for the private selection become tighter. 

\section{Distribution Specific Bounds for Private Selection} \label{sec:negbin}

We next consider privacy profile bounds for two specific choices of the distribution $K$, the truncated negative binomial distribution and the binomial distribution. As we show, in both cases the bounds allow evaluating a considerably larger number of private candidates than the state-of-the-art bounds. The bounds for the binomial distribution generalize the bounds for the Poisson distribution and improve the state-of-the-art bounds for the Poisson distribution and also allow concentrating $K$ further.

\subsection{Truncated Negative Binomial Distribution} \label{subsec:negbin}

Suppose the number of trials $K$ is distributed according to the truncated negative binomial distribution
$\mathcal{D}_{\eta,\gamma}$ which is determined by $\gamma \in (0,1)$ and $\eta \in (-1,\infty)$
and by the probabilities ($k \in \mathbb{N}$) for $\eta \neq 0$ by
$$
\mathbb{P}(K=k) = \frac{(1-\gamma)^k}{\gamma^{-\eta}-1} \prod_{i=0}^{k-1} \left( \frac{i + \eta}{i+1} \right).
$$
and for $\eta=0$ by 
$$
\mathbb{P}(K=k) = \frac{(1-\gamma)^k}{k \cdot \log{1/\gamma}}.
$$
It holds that when $\eta \neq 0$,
$$
\mathbb{E} K  = \frac{\eta (1-\gamma)}{\gamma ( 1-\gamma^{\eta})}
$$
and when $\eta=0$,
$$
\mathbb{E} K  = \frac{1/\gamma - 1}{\log 1/\gamma}.
$$
The derivative of the corresponding probability generating function is  given by
\begin{equation} \label{eq:negbin_varphi}
\varphi'(z) = \big( 1 - ( 1- \gamma)z \big)^{-\eta-1} \cdot \gamma^{\eta+1} \cdot \mathbb{E} K.
\end{equation}

As a baseline, we consider the following RDP bound.
\begin{theorem}[\citealt{papernot2022hyperparameter}]  \label{thm:rdp_negbin}
Let $Q$  satisfy $\big(\alpha,\epsilon \big)$-RDP and $\big(\wh\alpha,\wh\epsilon \big)$-RDP for some $\alpha \in (1,\infty)$ and
 $\wh\alpha \in [1,\infty)$. 
Draw $K$ from a truncated negative binomial distribution distribution $\mathcal{D}_{\eta,\gamma}$, where $\gamma \in (0,1)$ and $\eta \in (-1,\infty)$.
Run $Q(X)$ for $K$ times. Then $A(X)$ returns the best value of those $K$ runs 
(also $Q$'s output). 
Then $A$ satisfies $\big(\alpha,\epsilon'(\alpha) \big)$-RDP, where
\begin{equation*}
    \begin{aligned}
        \epsilon'(\alpha) = \epsilon(\alpha) + (\eta+1) \left( 1 - \frac{1}{\wh\alpha} \right) \wh \epsilon + \frac{(1+\eta) \cdot \log(1/\gamma)}{\hat \alpha} + \frac{\log m}{\alpha - 1}.
    \end{aligned}
\end{equation*}
\end{theorem}

Using the PGF~\eqref{eq:negbin_varphi} and our general result Thm.\ref{thm:main_thm}, we obtain the following bound using the hockey-stick divergence.

\begin{theorem} \label{thm:negbin}
Let $K \sim \mathcal{D}_{\eta,\gamma}$ and let $\delta(\epsilon_1)$, ${\epsilon_1} \in \mathbb{R}$, define the privacy profile of the base mechanism $Q$.
Then, for $A$ and $A'$, the output distributions of the selection algorithm evaluated on neighboring datasets $X$ and $X'$, respectively, and
for all $\epsilon_1 \geq 0$,
$$
H_{\ee^\epsilon}(A || A') \leq m \cdot \delta(\wh\epsilon)
$$
where 
\begin{equation*} 
 \wh\epsilon = \epsilon - (\eta+1)  \log \left( \ee^{\epsilon_1} + \frac{1-\gamma}{\gamma} \cdot {\delta}(\epsilon_1) \right).   
\end{equation*}
\end{theorem}

We directly obtain the following pure $\epsilon$-DP result from Thm.~\ref{thm:negbin}.
Notice that it includes Theorem 1.3 ($3 \epsilon$-bound) by~\citet{liu2019private} as a special case.

\begin{corollary} \label{cor:pointwise0}
Let $K \sim \mathcal{D}_{\eta,\gamma}$. If the base mechanism $Q$ is $\epsilon$-DP, then the selection algorithm $A$ is $ (\eta+2) \epsilon$-DP. For $\eta=1$ we get Theorem 1.3 of~\citep{liu2019private}.
\end{corollary}
The following $(\epsilon,\delta)$-DP result is also a straightforward corollary of Thm.~\ref{thm:negbin}.
\begin{corollary} \label{cor:pointwise}
Let $K \sim \mathcal{D}_{\eta,\gamma}$. If the base mechanism $Q$ is $(\epsilon,\delta)$-DP, then then the selection algorithm $A$ is $ \big((\eta+2) \epsilon + \gamma^{-1}  \delta, m \delta \big)$-DP.
\end{corollary}
Notice that for the geometric distribution ($\eta=1$), Cor.~\ref{cor:pointwise} implies that
if $Q$ is $(\epsilon,\delta)$-DP, then $A$ is $ \big(3 \epsilon + m \, \delta, m \, \delta \big)$-DP.

We can show that for a fixed $\delta$ the $\epsilon$-value of the private selection algorithm is
$\mathcal{O}(\log^{\frac{1}{2}} \frac{m}{\delta})$, in case the base mechanism is Gaussian differentially private~\citep{dong2022gaussian}.

\begin{corollary}[$\epsilon$-values when $Q$ is GDP] \label{lem:GDP}
Let $K \sim \mathcal{D}_{\eta,\gamma}$. 
Suppose the base mechanism is dominated by the Gaussian     mechanism with noise parameter $\sigma > 0$ and $L_2$-sensitivity 1. Then, for a fixed $\delta>0$, the private selection algorithm $A$ is $(\epsilon,\delta)$-DP for
$$
\epsilon = (\eta+2) \left( \frac{1}{2 \sigma^2} + \frac{1}{\sigma} \sqrt{ 2 \log \frac{1}{\gamma \cdot \delta}}\right) + \delta.
$$
\end{corollary}

Figure~\ref{fig:comparison_negbin} illustrates the various bounds for the truncated negative binomial distribution with $\eta=1$
when the base mechanism is the Gaussian mechanism with $\sigma=4$ and $L_2$-sensitivity 1, 
when $m =30,300$ and $3000$.
Figure~\ref{fig:comparison_negbin2} shows the increase of the $\epsilon$-values for a fixed %
value of $\delta$ and the bound of Lemma~\ref{lem:GDP} for comparison.

\begin{figure} [h!]
     \centering
        \includegraphics[width=.45\textwidth]{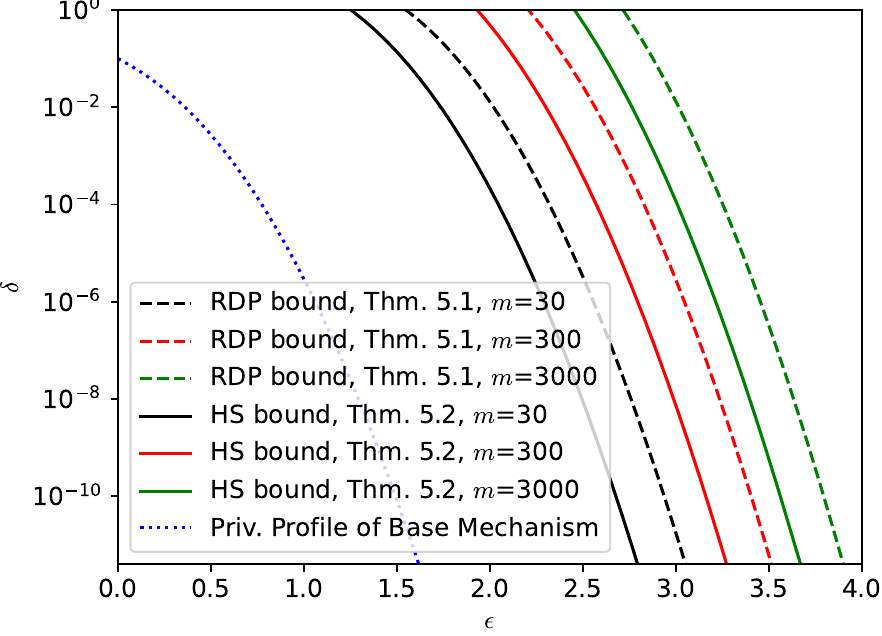}
        \caption{Comparison of various $(\epsilon,\delta)$-bounds when $K \sim \mathcal{D}_{\eta,\gamma}$ with $\eta=1$ (the geometric distribution, $m = \gamma^{-1}$) and $m=30,300,3000$. The base mechanism is the Gaussian mechanism with $L_2$-sensitivity 1 and noise parameter $\sigma=4$. 
        }
 	\label{fig:comparison_negbin}
\end{figure}

\begin{figure} [h!]
     \centering
        \includegraphics[width=.45\textwidth]{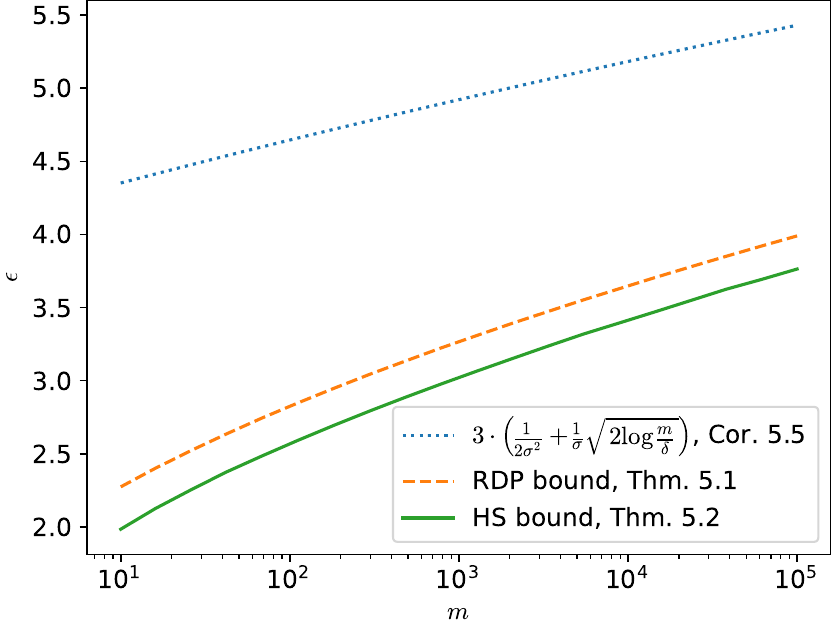}
        \caption{Growth of $\epsilon$-values for $\delta=10^{-6}$ for the RDP and hockey stick divergence based bounds as a function of $m$, when $K \sim \mathcal{D}_{\eta,\gamma}$ with $\eta=1$ and the base mechanism is the Gaussian mechanism with $L_2$-sensitivity 1 and noise parameter $\sigma=4$.
        The privacy profile bound given by Thm.~\ref{thm:negbin} retains the $\mathcal{O}(\log^{\frac{1}{2}} \frac{m}{\delta})$ growth of $\epsilon$-values from the DP RNM (See Fig.~\ref{fig:comparison_argmax2}).
        }
 	\label{fig:comparison_negbin2}
\end{figure}

\subsection{Binomial Distribution} \label{subsec:bin}

For choosing the distribution for $K$, we consider another alternative  that is suitable for smaller number of candidates.~\citet{papernot2022hyperparameter} consider as an alternative for the distribution of $K$
the Poisson distribution which is more concentrated than the truncated negative binomial distribution and thus gives an attractive choice from the practical point of view.
By considering $K$ to be binomially distributed, $K \sim \mathrm{Bin}(n,p)$, we can still further concentrate the distribution of $K$, with a small additional privacy cost. 
The hockey stick-based approach simplifies the analysis such that when compared to the RDP bounds for the case where $K$ is Poisson distributed, we essentially get much more concentrated $K$ for the same privacy cost using the binomial distribution.

Our bound for $K \sim \mathrm{Bin}(n,p)$ is a strict generalization of the Poisson distribution case, as we get the result for the Poisson distribution as a limit when $p \rightarrow 0$. The RDP bound to compare is the one given by~\citet{papernot2022hyperparameter}.

\begin{theorem}[\citealt{papernot2022hyperparameter}]  \label{thm:rdp_poisson}

Let $Q$  satisfy $\big(\alpha,\epsilon \big)$-RDP and $(\widehat{\epsilon},\widehat{\delta})$-DP for some $\alpha \in (1,\infty)$ and
$\epsilon,\widehat{\epsilon},\widehat{\delta} \geq 0$. 
Draw $K$ from a Poisson distribution with mean $m$.
Run $Q(X)$ for $K$ times. Then $A(X)$ returns the best value of those $K$ runs (also $Q$'s output).
If $K=0$, $A(X)$ returns some arbitrary output. If $\ee^{\widehat{\epsilon}} \leq 1 + \frac{1}{\alpha - 1}$, then $A$ satisfies $\big(\alpha,\epsilon'(\alpha) \big)$-RDP, where
$$
\epsilon'(\alpha) = \epsilon + m \cdot \widehat{\delta} + \frac{\log m}{\alpha - 1}.
$$
\end{theorem}

For a hockey-stick divergence bound, we consider the binomially distributed $K$ which is a strict generalization of the Poisson case and includes its privacy profile bound as a special case.
We can derive the following result from Thm.~\ref{thm:main_thm} when using the PGF of the binomial distribution. 
\begin{theorem} \label{thm:hs_bin}
Let $K \sim \mathrm{Bin}(n,p)$ for some $n \in \mathbb{N}$ and $0< p < 1$, and let $\delta({\epsilon_1})$, ${\epsilon_1} \in \mathbb{R}$, define the privacy profile of the base mechanism $Q$.
Suppose 
$$
{\epsilon_1} \geq \log \left( 1 + \frac{p}{1-p} \cdot \delta(\epsilon_1) \right).
$$
Then, for $A$ and $A'$, the output distributions of the selection algorithm evaluated on neighboring datasets $X$ and $X'$, respectively, for all $\epsilon > 0$ and
for all $\epsilon_1 \geq 0$,
\begin{equation} \label{eq:bin_bound_0}
    H_{\ee^\epsilon}(A || A') \leq m \cdot \delta(\wh{\epsilon}), 
\end{equation}
where 
$$
\wh{\epsilon} = \epsilon - (n-1) \log \big(  1 +  p \cdot (\ee^{\epsilon_1} - 1) + p \cdot \delta(\epsilon_1)  \big).
$$
\end{theorem}

We get the hockey-stick divergence bound for the case $K \sim \mathrm{Poisson}(m)$ as a corollary of Thm.~\ref{thm:hs_bin}.

\begin{corollary} \label{cor:binomial_to_poisson}
Let $K \sim \mathrm{Poisson}(m)$ for some $m \in \mathbb{N}$, and let $\delta({\epsilon_1})$, ${\epsilon_1} \in \mathbb{R}$, define the privacy profile of the base mechanism $Q$.
Then, for $A$ and $A'$, the output distributions of the selection algorithm evaluated on neighboring datasets $X$ and $X'$, respectively, and
for all $\epsilon > 0$ and for all $\epsilon_1 \geq 0$,
\begin{equation} \label{eq:poisson_bound_0}
    H_{\ee^\epsilon}(A || A') \leq m \cdot \delta(\wh{\epsilon}), 
\end{equation}
where 
$$
\wh{\epsilon} = \epsilon - m \cdot ( \ee^{\epsilon_1} - 1 ) - m \cdot \delta(\epsilon_1).
$$    
\end{corollary}
The proof of Cor.~\ref{cor:binomial_to_poisson} essentially follows from the fact that $\mathrm{Bin}(n,m/n)$ approaches $\mathrm{Poisson}(m)$ in total variation distance  as $n$ grows and that the bound~\eqref{eq:bin_bound_0} approaches the bound~\eqref{eq:poisson_bound_0} as $n$ grows. 
Figure~\ref{fig:comparison_bin} illustrates that when compared to the RDP bound of Thm.~\ref{thm:rdp_poisson} for the Poisson distributed $K$, 
we can obtain much smaller probabilities for small values of $K$ for the same privacy cost when using $K \sim \mathrm{Bin}(n,m/n)$ and Thm.~\ref{thm:hs_bin}.

\begin{figure} [h!]
     \centering
        \includegraphics[width=.47\textwidth]{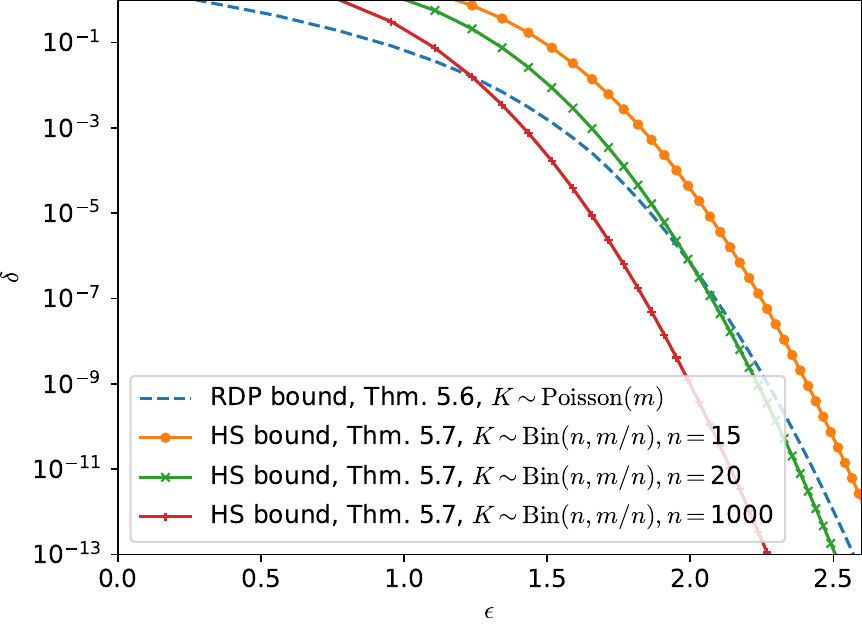}
        \includegraphics[width=.47\textwidth]{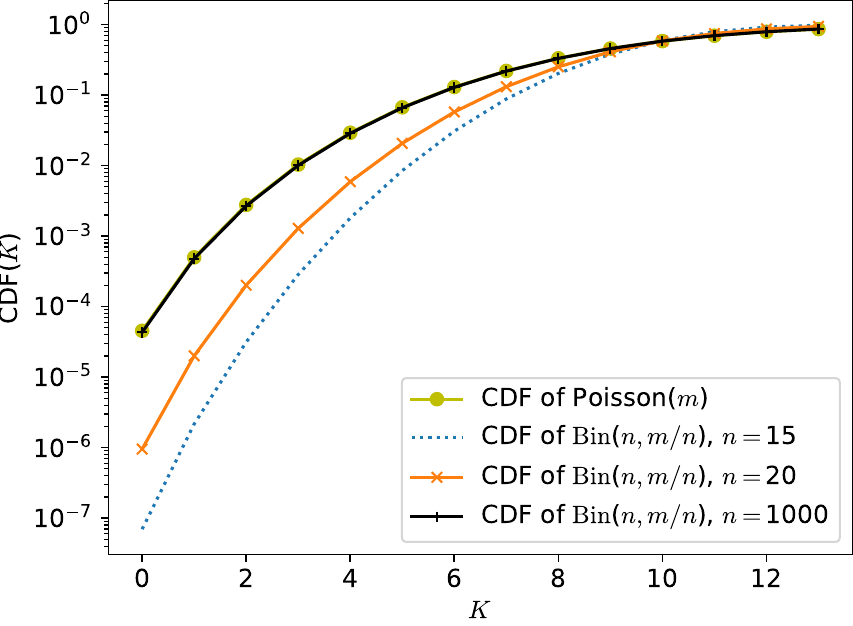}
        \caption{Top: Comparison of the bound of Thm.~\ref{thm:hs_bin} for $K \sim \textrm{Bin}(n,m/n)$ for different values of $n$, when $m=10$, and the RDP bound of Thm.~\ref{thm:rdp_poisson}. The base mechanism is the Gaussian mechanism with sensitivity 1 and $\sigma=4.0$. 
        Bottom: Comparison of the CDFs for different values of $n$. When comparing to the RDP bound, we see that at $\delta \approx 10^{-6}$ we get more concentrated $K$ for free by using the binomial distribution and Thm.~\ref{thm:hs_bin}.}
 	\label{fig:comparison_bin}
\end{figure}

\section{Applications} 

\subsection{Hyperparameter tuning for Propose-Test-Release}

Propose-Test-Release (PTR)~\citep{dwork2009differential} is one of the most versatile recipes for data-adaptive DP mechanism design. In its vanilla form, it involves three steps: (1) propose (or guess) a bound of the local sensitivity; (2) privately test whether the bound is valid; (3) If it passes the test, calibrate noise proportional to the proposed bound; otherwise, refuse to answer. Recently,~\citet{redberg2023generalized} generalized this approach by considering data-adaptive privacy losses instead.
The biggest challenge to apply the method is to know which bound to propose. The data-dependent privacy loss will depend on both the dataset and the hyperparameters of the query (as an example, think of the noise scale in additive noise mechanisms). To tune these hyperparameters we consider the private selection algorithm with geometrically distributed $K$.
The tricky issue with PTR is that PTR does not satisfy R\'enyi DP, and thus disqualifies the approach from \citet{papernot2022hyperparameter}.  Meanwhile, our methods deal with $(\varepsilon,\delta)$-DP and $\delta$-approximate Gaussian DP very naturally. Using our Thm.~\ref{thm:negbin} we can select the best threshold to propose a large number of candidates without resorting to composition. 
We have the following result for  Generalized PTR with an $(\wh\epsilon,\wh\delta)$-DP test.
We refer to~\citet{redberg2023generalized} for more details.
 \begin{theorem}[\citealt{redberg2023generalized}] \label{thm:gptr}
Consider a proposal $\phi$ and a data-dependent function $\epsilon_\phi(X)$ w.r.t. $\delta>0$. Suppose that we have an 
$(\wh\epsilon,\wh\delta)$-DP test $\mathcal{T} : \mathcal{X} \rightarrow \{0,1\}$ such that when $\epsilon_\phi(X) > \epsilon$,
$\mathcal{T}(X) = 1$ with probability $\delta'$ and 0 with probability $1-\delta'$. Then the Generalized PTR algorithm~\citep[][Alg.\;2]{redberg2023generalized} is $(\epsilon+\wh\epsilon,\delta+\wh\delta+\delta')$-DP.
\end{theorem}

Our approach is to wrap the private selection algorithm around Generalized PTR, and tune the parameter $\phi$. We use the point-wise guarantees given by Thm.~\ref{thm:gptr} for Gen. PTR and our Thm.~\ref{thm:rdp_negbin} for the tuning algorithm. To illustrate the effectiveness of this approach, we consider a linear regression problem on two UCI benchmark data sets~\citep{UCI}, see Fig.~\ref{fig:comparison_OPS}.
We apply  Generalized PTR to the one-posterior sample (OPS) algorithm described in~\citep{redberg2023generalized} which includes privately releasing the $L_2$-norm of the non-private solution and also the smallest eigenvalue of the feature covariance matrix.
As baselines we have the same approach using the privacy bounds of~\citet[Thm.\;3.5,][]{liu2019private},
the output perturbation method~\citep{Chaudhuri2011} and the non-adaptive method OPS-Balanced by~\citet{wang2018revisiting}.


\begin{figure}  [ht!]
     \centering
        \includegraphics[width=.45\textwidth]{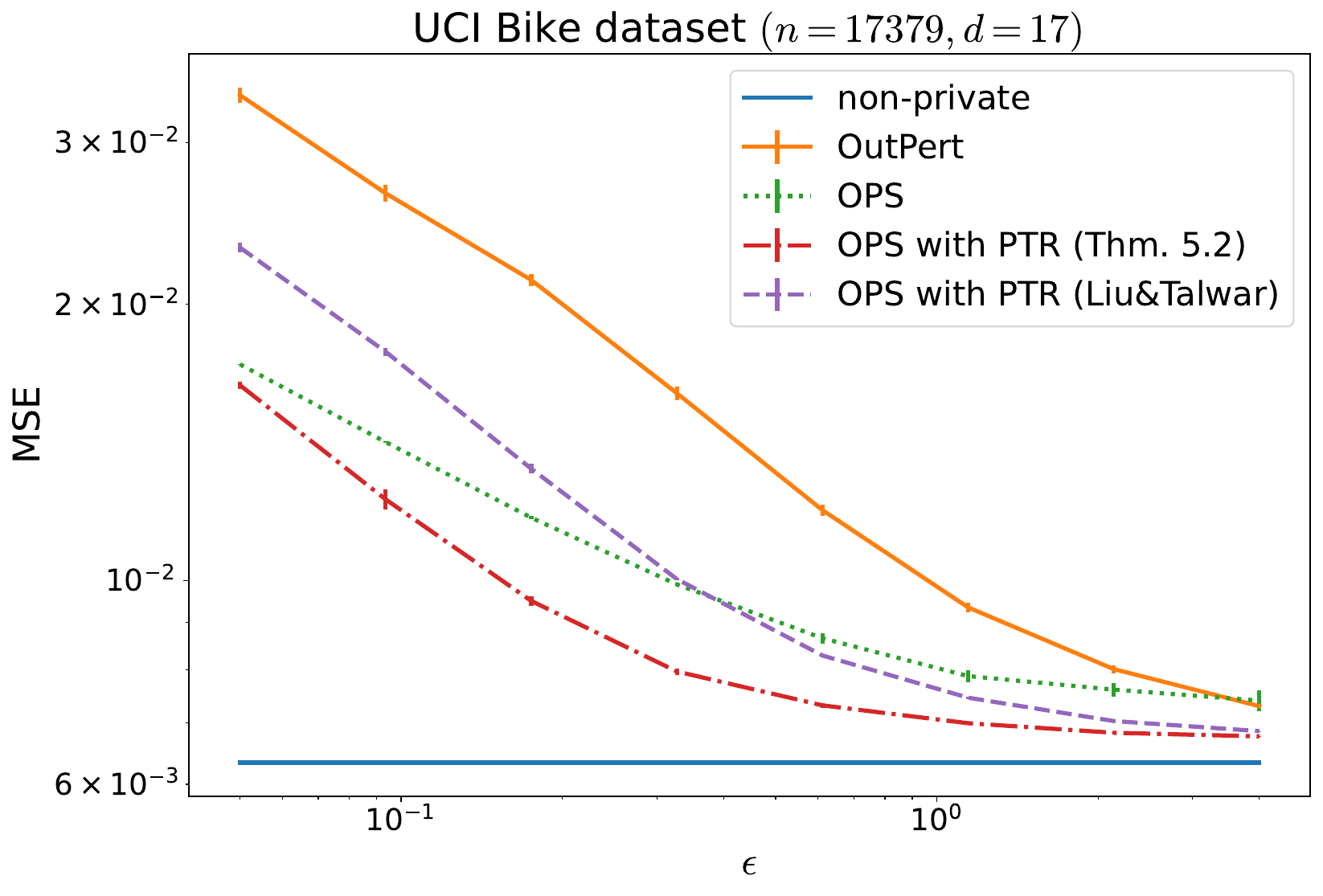}
        \includegraphics[width=.45\textwidth]{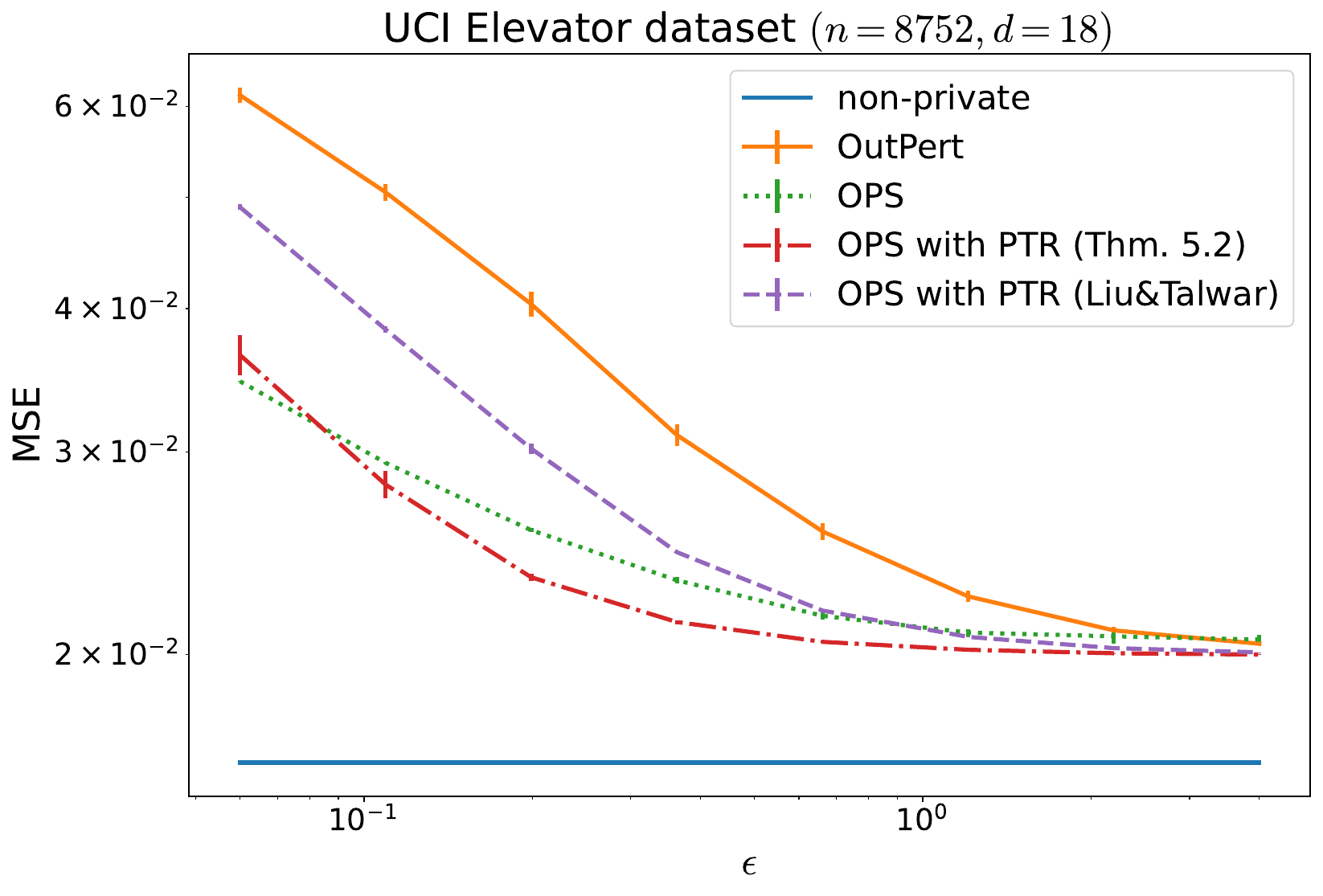}
        \caption{Linear regression problem on two UCI benchmark datasets. Tuning OPS-PTR (i.e., generalized PTR applied to the one-posterior sample algorithm) via the private selection algorithm outperforms baseline methods when the privacy cost of the tuning procedure is calculated using our Thm.~\ref{thm:negbin}.}
 	\label{fig:comparison_OPS}
\end{figure}

\subsection{Private Hyperparameter Tuning of DP-SGD}

Our results enable using numerical accountants for computing the privacy profile $\delta(\epsilon)$ for the subsampled Gaussian mechanism~\citep[see, e.g.,][]{koskela2021tight,gopi2021,zhu2021optimal}. We consider the simplest case, i.e., the Poisson subsampling and the add/remove neighborhood relation of datasets. We apply the numerical method proposed by~\citet{koskela2021tight} to the dominating pairs given in~\citep[Thm.\;11,][]{zhu2021optimal} to obtain accurate privacy profiles for the base mechanism. We remark that~\citet{zhu2021optimal} give privacy profiles for various subsampling schemes under both add/remove and substitute neighborhood relations of datasets. With these results, one can numerically construct the  dominating pair of distributions using the methods of~\citet{doroshenko2022connect} and also obtain upper bounds for compositions as well. 



We find that our bounds are tighter than the RDP bounds across a variety of parameter combinations.
Figure~\ref{fig:dpsgd1} shows comparisons with parameters taken from an example of~\citep{papernot2022hyperparameter}. RDP parameters are evaluated using the Opacus library~\citep{opacus}. Often, using larger batch sizes and noise ratios leads to increased privacy-utility tradeoff~\citep{de2022unlocking,ponomareva2023dp}. Figure~\ref{fig:dpsgd2} shows comparisons in a setting of such a high-accuracy experiment~\citep{de2022unlocking}.

\begin{figure} [ht!]
     \centering
        \includegraphics[width=.45\textwidth]{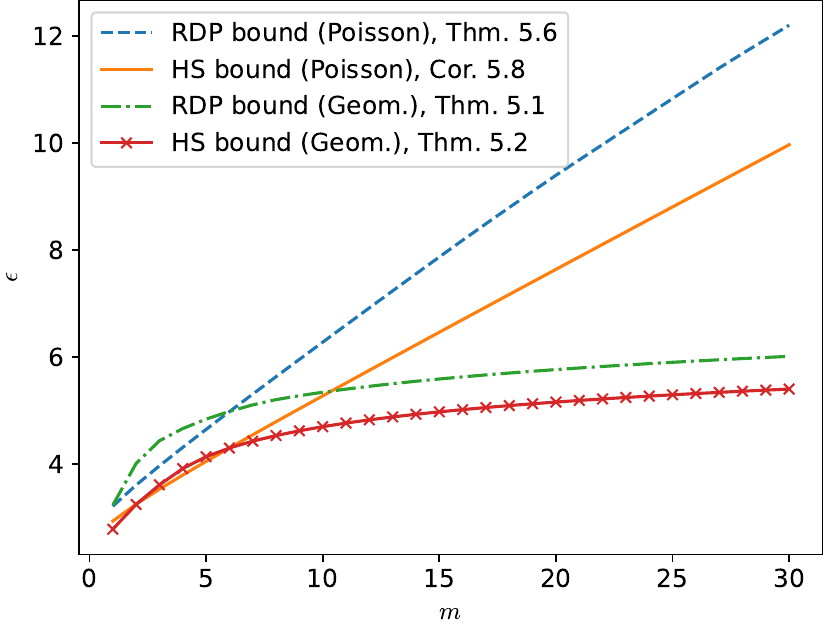}
        \caption{Comparison of the hockey-stick divergence bounds of Thm.~\ref{thm:negbin} and Cor.~\ref{cor:binomial_to_poisson}, and the RDP bounds of Thm.~\ref{thm:rdp_negbin} and~\ref{thm:rdp_poisson} for the private selection for a Poisson-subsampled Gaussian mechanism with   subsampling ratio $q=256/60000$, noise parameter $\sigma=1.1$ and number of steps $T= \ceil[]{60/q }$.}
 	\label{fig:dpsgd1}
\end{figure}

\begin{figure} [ht!]
     \centering
        \includegraphics[width=.45\textwidth]{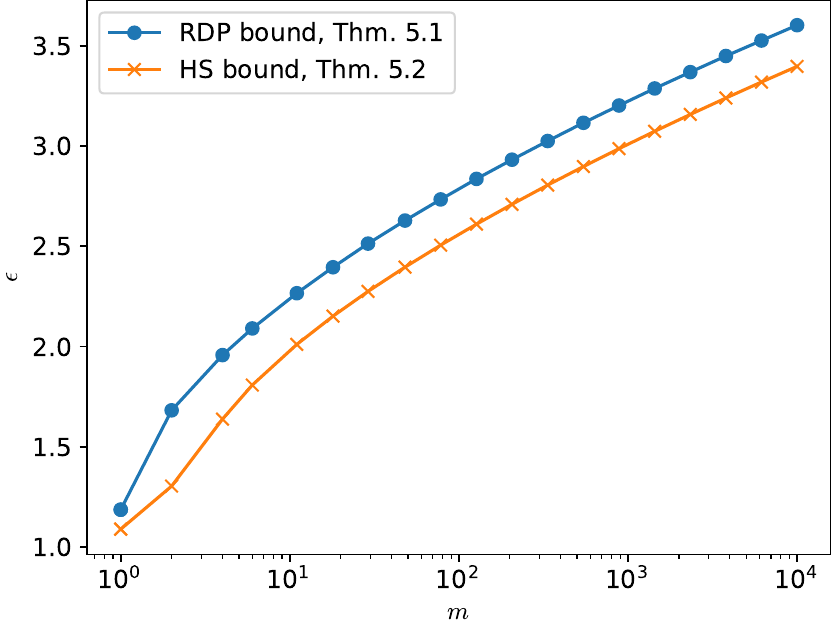}
        \caption{Comparison of the bounds, when the base mechanism is a Poisson-subsampled Gaussian mechanism with the parameters $q=16384/50000$, $\sigma=21.1$ and $T=250$~\citep[see Table 18 in][]{de2022unlocking}. We see that the bound of Thm.~\ref{thm:negbin} allows evaluating approximately 3 times as many models as the RDP bound.}
 	\label{fig:dpsgd2}
\end{figure}

\paragraph{Adjusting Hyperparameters.} 
One important question is, how to adjust the hyperparameters in case we are optimizing parameters that affect the privacy guarantees themselves. For example, one may consider tuning the noise parameter $\sigma$ in DP-SGD in which case one may also consider adjusting the length of the training. 

The bound of Thm.~\ref{thm:negbin} essentially requires only point-wise information about the privacy profile of the base mechanism $Q$. However, the bound can be optimized w.r.t. the privacy profile of $Q$ which may affect it considerably. 
For finding ideal point-wise DP thresholds, we consider the following procedure. Suppose the base mechanisms all satisfy  $(\epsilon_Q,\delta)$ for some $\epsilon_Q>0$ and $\delta>0$. As the privacy profiles start to resemble those of the Gaussian mechanism for large numbers of compositions~\citep{dong2022gaussian}, we carry out the optimization of the upper bound of Thm.~\ref{thm:negbin} w.r.t. to the privacy profile of a Gaussian mechanism (GM) that is adjusted to be $(\epsilon_Q,\delta)$-DP. More specifically, we first adjust the noise scale $\sigma$ of the GM such that it is $(\epsilon_Q,\delta)$-DP, to obtain a privacy profile $\delta_Q(\epsilon)$. Then, in case we are using $K \sim \mathcal{D}_{\eta,\gamma}$, we carry out the minimization  
$$
\epsilon_1 = \argmin_{\epsilon} \log\left( \ee^{\epsilon} + \frac{1-\gamma}{\gamma} \delta_Q(\epsilon) \right).
$$
and set $\delta_1 = \delta_Q(\epsilon_1)$. This results in the first threshold $(\epsilon_1,\delta_1)$.
Using the same privacy profile we find an $\wh \epsilon$-value where this GM is $(\wh \epsilon,\delta/m)$-DP. The following  corollary result of Thm.~\ref{thm:negbin} then gives the
$\epsilon$-value for which the private selection algorithm is $(\epsilon,\delta)$-DP.  
\begin{corollary}  \label{cor:adjust}
    Suppose the base mechanism $Q$ is $(\epsilon_1,\delta_1)$-DP and  $(\wh \epsilon,\delta/m)$-DP for some $\epsilon_1\geq0$ and $\wh \epsilon \geq 0$. Then, the private selection algorithm with $K \sim \mathcal{D}_{\eta,\gamma}$ is $(\epsilon,\delta)$-DP for 
    $$
    \epsilon = \wh \epsilon + (\eta+1) \log\left( \ee^{\epsilon_1} + \frac{1-\gamma}{\gamma} \delta_1 \right).
    $$
\end{corollary}
Figure~\ref{fig:DPSGD_adjusted} shows the upper bound obtained using this procedure, when we are tuning the $\sigma$-parameter for the Poisson-subsampled Gaussian mechanism. We fix $q=0.01$ and set as a threshold $\epsilon_Q=1.5$ and $\delta=10^{-6}$.
We consider three noise-level candidates: $\sigma=2.0$, $\sigma=3.0$ and $\sigma=4.0$. For each of these three candidates, the number of iterations $T$ is determined to be the maximum such that the privacy profile of the candidate is below $(\epsilon_1,\delta_1)$- and $(\wh \epsilon,\delta/m)$-thresholds. As a result we can run the candidate models for $\approx$ 4000, 10000 and 18000 iterations, respectively. We compare the $(\epsilon,\delta)$-DP bound given by Cor.~\ref{cor:adjust} and the bounds we would obtain by optimizing Thm.~\ref{thm:negbin} individually for each candidate mechanism and see that there is a very small gap.

\begin{figure} [H]
     \centering
        \includegraphics[width=.35\textwidth]{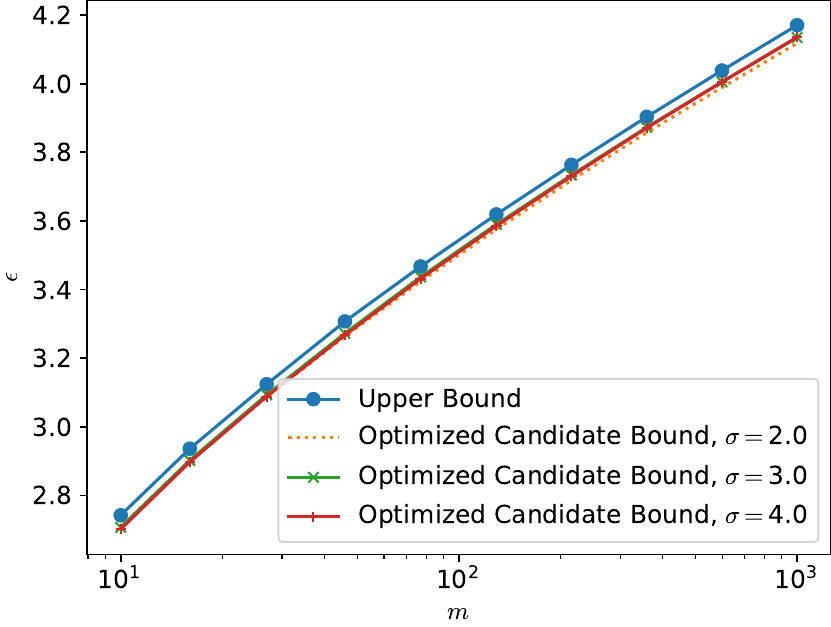}
        \caption{Tuning the $\sigma$-parameter for the Poisson-subsampled Gaussian mechanism with geometrically-distributed $K$. We fix $q=0.01$ and consider three $\sigma$-candidates: 2.0, 3.0 and 4.0. Shown are the $(\epsilon,\delta)$-bounds given by Cor.~\ref{cor:adjust} and the bounds obtained by optimizing Thm.~\ref{thm:negbin} individually for each candidate mechanism.}
 	\label{fig:DPSGD_adjusted}
\end{figure}




\section{Conclusions} 

We have filled a gap in the private selection literature by providing $(\epsilon,\delta)$-privacy analysis for various selection algorithms that is able to accurately use the privacy profiles of the candidate mechanisms. When compared to existing RDP bounds, in DP-SGD tuning, for example, the new bounds allow evaluating approximately 3 times as many candidate models.
The bounds also improve existing point-wise $(\epsilon,\delta)$-bounds which translates to improved utility in data-adaptive
analyses using the generalized PTR framework. We have also shown how to use the bounds to adjust parameters of the candidate models when tuning hyperparameters that affect the privacy guarantees of the candidate mechanisms. 

\subsection*{Acknowledgments}
The work was partially supported by NSF Award \#2048091.

\bibliography{priv_sel}

\begin{thebibliography}{}

\bibitem[Abadi et~al., 2016]{Abadi2016}
Abadi, M., Chu, A., Goodfellow, I., McMahan, H.~B., Mironov, I., Talwar, K.,
  and Zhang, L. (2016).
\newblock Deep learning with differential privacy.
\newblock In {\em Proceedings of the 2016 ACM SIGSAC Conference on Computer and
  Communications Security}, pages 308--318.

\bibitem[Bache and Lichman, 2013]{UCI}
Bache, K. and Lichman, M. (2013).
\newblock {UCI} machine learning repository.

\bibitem[Balle et~al., 2018]{balle2018subsampling}
Balle, B., Barthe, G., and Gaboardi, M. (2018).
\newblock Privacy amplification by subsampling: Tight analyses via couplings
  and divergences.
\newblock {\em Advances in Neural Information Processing Systems}, pages
  6277--6287.

\bibitem[Balle and Wang, 2018]{balle2018improving}
Balle, B. and Wang, Y.-X. (2018).
\newblock Improving the {G}aussian mechanism for differential privacy:
  Analytical calibration and optimal denoising.
\newblock In {\em International Conference on Machine Learning}, pages
  394--403. PMLR.

\bibitem[Boyd and Vandenberghe, 2004]{boyd2004convex}
Boyd, S.~P. and Vandenberghe, L. (2004).
\newblock {\em Convex optimization}.
\newblock Cambridge university press.

\bibitem[Canonne et~al., 2020]{canonne2020discrete}
Canonne, C.~L., Kamath, G., and Steinke, T. (2020).
\newblock The discrete gaussian for differential privacy.
\newblock {\em Advances in Neural Information Processing Systems},
  33:15676--15688.

\bibitem[Chaudhuri et~al., 2011]{Chaudhuri2011}
Chaudhuri, K., Monteleoni, C., and Sarwate, A.~D. (2011).
\newblock Differentially private empirical risk minimization.
\newblock {\em J. Mach. Learn. Res.}, 12:1069--1109.

\bibitem[Chaudhuri and Vinterbo, 2013]{chaudhuri2013stability}
Chaudhuri, K. and Vinterbo, S.~A. (2013).
\newblock A stability-based validation procedure for differentially private
  machine learning.
\newblock {\em Advances in Neural Information Processing Systems}, 26.

\bibitem[De et~al., 2022]{de2022unlocking}
De, S., Berrada, L., Hayes, J., Smith, S.~L., and Balle, B. (2022).
\newblock Unlocking high-accuracy differentially private image classification
  through scale.
\newblock {\em arXiv preprint arXiv:2204.13650}.

\bibitem[Dong et~al., 2022]{dong2022gaussian}
Dong, J., Roth, A., and Su, W.~J. (2022).
\newblock Gaussian differential privacy.
\newblock {\em Journal of the Royal Statistical Society Series B}, 84(1):3--37.

\bibitem[Doroshenko et~al., 2022]{doroshenko2022connect}
Doroshenko, V., Ghazi, B., Kamath, P., Kumar, R., and Manurangsi, P. (2022).
\newblock Connect the dots: Tighter discrete approximations of privacy loss
  distributions.
\newblock {\em arXiv preprint arXiv:2207.04380}.

\bibitem[Dwork, 2006]{dwork06differential}
Dwork, C. (2006).
\newblock Differential privacy.
\newblock In {\em Proc. 33rd Int. Colloq. on Automata, Languages and Prog.
  (ICALP 2006), Part {II}}, pages 1--12.

\bibitem[Dwork and Lei, 2009]{dwork2009differential}
Dwork, C. and Lei, J. (2009).
\newblock Differential privacy and robust statistics.
\newblock In {\em Proceedings of the forty-first annual ACM symposium on Theory
  of computing}, pages 371--380.

\bibitem[Gopi et~al., 2021]{gopi2021}
Gopi, S., Lee, Y.~T., and Wutschitz, L. (2021).
\newblock Numerical composition of differential privacy.
\newblock {\em Advances in Neural Information Processing Systems},
  34:11631--11642.

\bibitem[Koskela et~al., 2021]{koskela2021tight}
Koskela, A., J{\"a}lk{\"o}, J., Prediger, L., and Honkela, A. (2021).
\newblock Tight differential privacy for discrete-valued mechanisms and for the
  subsampled {G}aussian mechanism using {FFT}.
\newblock In {\em International Conference on Artificial Intelligence and
  Statistics}, pages 3358--3366. PMLR.

\bibitem[Liese and Vajda, 2006]{liese2006divergences}
Liese, F. and Vajda, I. (2006).
\newblock On divergences and informations in statistics and information theory.
\newblock {\em IEEE Transactions on Information Theory}, 52(10):4394--4412.

\bibitem[Liu and Talwar, 2019]{liu2019private}
Liu, J. and Talwar, K. (2019).
\newblock Private selection from private candidates.
\newblock In {\em Proceedings of the 51st Annual ACM SIGACT Symposium on Theory
  of Computing}, pages 298--309.

\bibitem[Mironov, 2017]{mironov2017}
Mironov, I. (2017).
\newblock R\'enyi differential privacy.
\newblock In {\em 2017 IEEE 30th Computer Security Foundations Symposium
  (CSF)}, pages 263--275.

\bibitem[Mohapatra et~al., 2022]{mohapatra2022role}
Mohapatra, S., Sasy, S., He, X., Kamath, G., and Thakkar, O. (2022).
\newblock The role of adaptive optimizers for honest private hyperparameter
  selection.
\newblock In {\em Proceedings of the AAAI Conference on Artificial
  Intelligence}, volume~36, pages 7806--7813.

\bibitem[Papernot and Steinke, 2022]{papernot2022hyperparameter}
Papernot, N. and Steinke, T. (2022).
\newblock Hyperparameter tuning with renyi differential privacy.
\newblock In {\em International Conference on Learning Representations}.

\bibitem[Ponomareva et~al., 2023]{ponomareva2023dp}
Ponomareva, N., Hazimeh, H., Kurakin, A., Xu, Z., Denison, C., McMahan, H.~B.,
  Vassilvitskii, S., Chien, S., and Thakurta, A.~G. (2023).
\newblock How to dp-fy ml: A practical guide to machine learning with
  differential privacy.
\newblock {\em Journal of Artificial Intelligence Research}, 77:1113--1201.

\bibitem[Redberg et~al., 2023]{redberg2023generalized}
Redberg, R., Zhu, Y., and Wang, Y.-X. (2023).
\newblock Generalized ptr: User-friendly recipes for data-adaptive algorithms
  with differential privacy.
\newblock In {\em International Conference on Artificial Intelligence and
  Statistics}, pages 3977--4005. PMLR.

\bibitem[Tang et~al., 2024]{tang2024privacy}
Tang, X., Shin, R., Inan, H.~A., Manoel, A., Mireshghallah, F., Lin, Z., Gopi,
  S., Kulkarni, J., and Sim, R. (2024).
\newblock Privacy-preserving in-context learning with differentially private
  few-shot generation.
\newblock In {\em International Conference on Learning Representations}.

\bibitem[Wang, 2018]{wang2018revisiting}
Wang, Y.-X. (2018).
\newblock Revisiting differentially private linear regression: optimal and
  adaptive prediction \& estimation in unbounded domain.
\newblock {\em Uncertainty in Artificial Intelligence (UAI-18)}.

\bibitem[Wu et~al., 2024]{wu2024privacy}
Wu, T., Panda, A., Wang, J.~T., and Mittal, P. (2024).
\newblock Privacy-preserving in-context learning for large language models.
\newblock In {\em International Conference on Learning Representations}.

\bibitem[Yousefpour et~al., 2021]{opacus}
Yousefpour, A., Shilov, I., Sablayrolles, A., Testuggine, D., Prasad, K.,
  Malek, M., Nguyen, J., Ghosh, S., Bharadwaj, A., Zhao, J., et~al. (2021).
\newblock Opacus: User-friendly differential privacy library in pytorch.
\newblock In {\em NeurIPS 2021 Workshop Privacy in Machine Learning}.

\bibitem[Zhu et~al., 2022]{zhu2021optimal}
Zhu, Y., Dong, J., and Wang, Y.-X. (2022).
\newblock Optimal accounting of differential privacy via characteristic
  function.
\newblock In {\em Proceedings of The 25th International Conference on
  Artificial Intelligence and Statistics}.

\bibitem[Zhu and Wang, 2022]{zhu2022adaptive}
Zhu, Y. and Wang, Y.-X. (2022).
\newblock Adaptive private-k-selection with adaptive k and application to
  multi-label pate.
\newblock In {\em International Conference on Artificial Intelligence and
  Statistics}, pages 5622--5635. PMLR.

\end{thebibliography}

\newpage
\appendix

\section{Proofs for Section~\ref{sec:max_sel}}

\subsection{Proof of Theorem~\ref{thm:max_sel} (Privacy Profile of Additive Noise RNM)}

\begin{theorem}
Let $X \sim X'$ and $\epsilon \in \mathbb{R}$. We have
$$
H_{\ee^\epsilon} \big(\mathcal{M}(X) || \mathcal{M}(X') \big) \leq m \cdot \delta(\epsilon),
$$
where $\delta(\epsilon)$ is the privacy profile of the additive noise mechanism with sensitivity 2. 
\end{theorem}
\begin{proof}
For each $i \in [m]$, denoting the density function of $Z_i$ by $p(r_i)$, we have that
\begin{equation} \label{eq:phase0}
\begin{aligned}
& H_f \big(\mathcal{M}(X) || \mathcal{M}(X') \big)    \\
&=  \sum\limits_{i=1}^m f \left( \frac{\mathbf{P}(\mathcal{M}(X) = i)}{\mathbf{P}(\mathcal{M}(X') = i)} \right)  \mathbf{P}(\mathcal{M}(X') = i) \\
& =  \sum\limits_{i=1}^m  
f \left( \frac{\int_{-\infty}^{\infty} p(r_i-2) \cdot \mathbf{P}(f_i(X) + r_i -2 > \max_{j\in [m], j\neq i} \{f_j(X)+r_j\}) \, \dd r_i }{\int_{-\infty}^{\infty} p(r_i) \cdot \mathbf{P}(f_i(X')  + r_i > \max_{j\in [m], j\neq i} \{f_j(X') +r_j\}) \, \dd r_i }   \right) \\
& \quad \quad \quad  \quad \cdot \int_{-\infty}^{\infty} p(r_i) \cdot \mathbf{P}(f_i(X')  + r_i > \max_{j\in [m], j\neq i} \{f_j(X') +r_j\}) \, \dd r_i,
\end{aligned}
\end{equation}
where $p(r_i)$'s are the density functions of $Z_i$'s, respectively, and
where the randomness in $\mathbf{P}(f_i(X) + r_i > \max_{j\in [m], j\neq i} \{f_j(X)+r_j\})$ is w.r.t. $r_j$'s.

From the Lipschitz property it follows that for all $i \in [m]$,
\begin{equation} \label{eq:lipschitz}
    \frac{\mathbf{P}(f_i(X) + r_i -2 > \max_{j\in [m], j\neq i} \{f_j(X)+r_j\})  }{\mathbf{P}(f_i(X')  + r_i > \max_{j\in [m], j\neq i} \{f_j(X') +r_j\})  }   \leq 1. 
\end{equation}
For an $f$ divergence determined by a convex $f$ of the RNM of additive noise mechanisms evaluated at $X$ and $X'$, we have that
\begin{equation} \label{eq:m_delta} 
\begin{aligned}
& H_f \big(\mathcal{M}(X) || \mathcal{M}(X') \big) =
\sum_{i=1}^m  f\left( \frac{\mathbf{P}(\mathcal{M}(X) = i)}{\mathbf{P}(\mathcal{M}(X') = i)} \right) \mathbf{P}(\mathcal{M}(X') = i)  \\
&=
\sum_{i=1}^m f \left( \frac{\int_{-\infty}^{\infty} p(r_i-2) \cdot \mathbf{P}(f_i(X) + r_i -2 > \max_{j\in [m], j\neq i} \{f_j(X)+r_j\}) \, \dd r_i }{\int_{-\infty}^{\infty} p(r_i) \cdot \mathbf{P}(f_i(X')  + r_i > \max_{j\in [m], j\neq i} \{f_j(X') +r_j\}) \, \dd r_i }   \right) \\
& \quad \quad \quad  \quad \cdot \int_{-\infty}^{\infty} p(r_i) \cdot \mathbf{P}(f_i(X')  + r_i > \max_{j\in [m], j\neq i} \{f_j(X') +r_j\}) \, \dd r_i \\
& \leq
\sum_{i=1}^m \int_{-\infty}^{\infty} f \left( \frac{ p(r_i-2) \cdot \mathbf{P}(f_i(X) + r_i -2 > \max_{j\in [m], j\neq i} \{f_j(X)+r_j\})  }{ p(r_i) \cdot \mathbf{P}(f_i(X')  + r_i > \max_{j\in [m], j\neq i} \{f_j(X') +r_j\}) }   \right) \\
& \quad \quad \quad  \quad \cdot p(r_i) \cdot \mathbf{P}(f_i(X')  + r_i > \max_{j\in [m], j\neq i} \{f_j(X') +r_j\}) \, \dd r_i \\
& \leq
\sum_{i=1}^m \int_{-\infty}^{\infty} f \left( \frac{ p(r_i-2)   }{ p(r_i) }   \right) \cdot p(r_i) \cdot \mathbf{P}(f_i(X')  + r_i > \max_{j\in [m], j\neq i} \{f_j(X') +r_j\}) \, \dd r_i \\
& \leq
\sum_{i=1}^m \int_{-\infty}^{\infty} f \left( \frac{ p(r_i-2)   }{ p(r_i) }   \right) \cdot p(r_i) \, \dd r_i \\
& =
\sum_{i=1}^m H_f \big(p(r_i) || p(r_i-2) \big) \\
& = m \cdot \delta(\epsilon),
\end{aligned}
\end{equation}
in case $f$ corresponds to the hockey-stick-divergence and
where $\delta(\epsilon)$ is the privacy profile of the additive noise mechanism with sensitivity 2. In the first inequality we have used  Lemma~\ref{lem:f_convex} and Jensen's inequality and in the second inequality we have used the Lipschitz property and the fact that $f(z)$ is a non-decreasing function of $z$.
In case of monotonicity, i.e., if $f_i(X) \geq f_i(X')$ for all $i \in [m]$, the condition~\eqref{eq:lipschitz} holds with $r_i-2$ replaced by $r_i-1$ and we have the result with $r_i - 2$ replaced by $r_i-1$.

\end{proof}

\subsection{Proof of Theorem~\ref{lem:composition_rnm} }
\begin{lemma}
In case of an adaptive composition of $k$ mechanisms of the form~\eqref{eq:M_argmax}, we get the 
privacy profile upper bound $m^k \cdot \delta(\epsilon)$, where $\delta(\epsilon)$ is the privacy profile of a $m$-wise composition of the additive noise mechanism with noise $Z$ and sensitivity $2$. 
\begin{proof}
We use the proof technique used in~\citep[Thm.\;27][]{zhu2021optimal}
and consider an adaptive composition of two mechanisms. The general case follows from the proof.
Let $X \sim X'$. Denote the density functions of $\mathcal{M}^1(X)$ and $\mathcal{M}^2\big(\mathcal{M}^1(X),X \big)$ by $f_1(t)$ and $f_2(t,s)$, respectively, and the density functions of $\mathcal{M}^1(X')$ and $\mathcal{M}^2\big(\mathcal{M}^1(X'),X'\big)$ by $f_1'(t)$ and $f_2'(t,s)$, respectively.
Denote the density function of $Z$ by $p(t)$. Using the bound given by Thm.~\ref{thm:max_sel}, we have:
\begin{equation*}
    \begin{aligned}
        H_{\ee^\epsilon} \big( (\mathcal{M}(X),\mathcal{M}(X) \big) || (\mathcal{M}(X),\mathcal{M}(X) \big)  \big)  & =
   \int_{\mathbb{R}} \int_{\mathbb{R}} \max \{ f_1(t) f_2(t,s) - \ee^\epsilon f_1'(t) f_2'(t,s), 0 \} \, \dd s \dd t \\
    &= \int_{\mathbb{R}} f_1(t) \left( \int_{\mathbb{R}} \max \{  f_2(t,s) - \ee^{\epsilon - \log \tfrac{f_1(t)}{f_1'(t)}} f_2'(t,s), 0 \} \, \dd s \right) \dd t \\
    &\leq m \cdot \int_{\mathbb{R}} f_1(t) \left( \int_{\mathbb{R}} \max \{  p(s-2) - \ee^{\epsilon - \log \tfrac{f_1(t)}{f_1'(t)}} p(s), 0 \} \, \dd s \right) \dd t \\
    &= m \cdot \int_{\mathbb{R}} p(s-2) \left( \int_{\mathbb{R}} \max \{  f_1(t) - \ee^{\epsilon - \log \tfrac{p(s-2)}{p(s)}} f_1'(t), 0 \} \, \dd t \right) \dd s \\
    &\leq m^2 \cdot \int_{\mathbb{R}} p(s-2) \left( \int_{\mathbb{R}} \max \{  p(t-2) - \ee^{\epsilon - \log \tfrac{p(s-2)}{p(s)}} p(t), 0 \} \, \dd t \right) \dd s \\
    &= m^2 \cdot \int_{\mathbb{R}} \int_{\mathbb{R}}  \max \{ p(s-2) p(t-2)   - \ee^{\epsilon} p(s) p(t), 0 \} \, \dd t \dd s \\
    \end{aligned}
\end{equation*}
which shows the claim for $k=2$. The general case follows by induction.
\end{proof}
\end{lemma}

\subsection{Proof of Corollary~\ref{lem:convert2_eps_delta} }

\begin{lemma} 
Consider the mechanism $\mathcal{M}$ defined in Eq.~\ref{eq:M_argmax} and suppose $Z$ is normally distributed with variance $\sigma^2$. 
Let $\delta>0$. Then $\mathcal{M}$ is $(\epsilon,\delta)$-DP for 
$$
\epsilon = \frac{2}{\sigma^2} + \frac{2}{\sigma} \sqrt{ 2 \log \frac{m}{\delta} }
$$
\begin{proof}
Let $\delta>0$. By~\citep[Lemma 3,][]{balle2018improving} we know that the Gaussian mechanism with $L_2$-sensitivity $\Delta$ and noise scale $\sigma$ is $(\epsilon,\delta)$-DP if 
$$
\mathbb{P}( \omega \geq \epsilon) \leq \delta,
$$
where
$$
\omega \sim \mathcal{N}\left( \frac{\Delta^2}{2 \sigma^2}, \frac{\Delta^2}{\sigma^2} \right).
$$
Using a simple Chernoff bound for the Gaussian, we see that for any $\epsilon \geq \tfrac{\Delta^2}{2 \sigma^2}$,
$$
\mathbb{P}( \omega \geq \epsilon) \leq  \ee^{ - \frac{\wt\epsilon^2 \sigma^2}{2}},
$$
where $\wt \epsilon = \epsilon -  \frac{\Delta^2}{2 \sigma^2}$. Setting 
$$
\epsilon = \tfrac{\Delta^2}{2 \sigma^2} + \tfrac{\Delta}{\sigma} \sqrt{ 2 \log \tfrac{m}{\delta} },
$$
we see that
$$
\mathbb{P}( \omega \geq \epsilon) \leq  \frac{1}{m}.
$$
The claim follows setting $\Delta=2$ and using Thm.~\ref{thm:max_sel}.

\end{proof}
\end{lemma}

\section{Proof of Theorem~\ref{thm:main_thm} (General Bound for the Privacy Selection)}

Similarly to~\citep[Lemma 7,][]{papernot2022hyperparameter} we denote
$Q(\leq y) = \sum_{y' \leq y} Q(y')$ and $Q(< y) = \sum_{y' < y} Q(y')$ and formulate
our main result for the case $\mathbb{P}(K=0)=0$. The case $\mathbb{P}(K=0)>0$ could be included in the upper bound as a small additional term of the form $f(0) \cdot \mathbb{P}(K=0)$. E.g., in case of the hockey stick divergence, it does not account for the divergence in case $\epsilon>0$ so we neglect it.

Also, $\varphi$ denotes the probability generating function of $K$, i.e.,
$$
\varphi(z) = \sum\limits_{k=1}^\infty \mathbb{P}(K=k) \cdot z^k.
$$
As shown in~\citep[Lemma 7,][]{papernot2022hyperparameter},
\begin{equation} \label{eq:RDP_A0} 
\begin{aligned}
A(y) &= \sum_{k=1}^{\infty} \mathbf{P}[K=k] \cdot \big( Q( \leq y)^k - Q(<y)^k \big) \\
 &=  \varphi \big( Q( \leq y) \big)  - \varphi \big(  Q(<y) \big) \\
 &= \int_{Q(<y)}^{Q( \leq y)} \varphi '(z) \, \dd z \\
& = Q(y) \cdot \mathop{\mathbb{E}}_{X \leftarrow [Q(<y),Q(\leq y)]} [\varphi'(X)],
\end{aligned}
\end{equation}
where $X \leftarrow [Q(<y),Q(\leq y)]$ denotes uniformly distributed r.v. on the interval $[Q(<y),Q(\leq y)]$.

\begin{theorem}
Let $X \sim X'$ and let $A$ and $A'$ be the density functions of the hyperparameter tuning algorithm, evaluated on $X$ and $X'$, respectively. Let $Q$ and $Q'$ be the density functions of the quality score of the base mechanism, evaluated on  $X$ and $X'$, respectively. Let $K$ be random variable for the times the base mechanism is run and $\varphi(z)$ the PGF of $K$.
Let $f \, : \,  [0,\infty) \rightarrow \mathbb{R}$ be a convex function.
Then, 
\begin{equation*} 
\begin{aligned}
H_f (A || A') 
 \leq \sum_{y \in \mathcal{Y}} 
f \left( \frac{ Q(y)  \varphi'(q_y)   }{ Q'(y) \varphi'(q'_y) }\right) \cdot Q'(y) \varphi'(q'_y),
\end{aligned}
\end{equation*}
where for each $y \in \mathcal{Y}$, $q_y$ and $q'_y$ are obtained by applying the same 
$y$-dependent post-processing function to
$Q$ and $Q'$, respectively.
\begin{proof}
For the mechanism $A$ defined in~\eqref{eq:RDP_A0}, we can bound the HS divergence as follows:
\begin{equation} \label{eq:H_alpha}
\begin{aligned}
H_f (A || A') &= \sum_{y \in \mathcal{Y}} f \left( \frac{ A(y)   }{ A'(y) }\right) \cdot A(y) \\
&= \sum_{y \in \mathcal{Y}} f \left( \frac{ Q(y) \mathop{\mathbb{E}}_{X \leftarrow [Q(<y),Q(\leq y)]} [\varphi'(X)]   }{ Q'(y) \mathop{\mathbb{E}}_{X' \leftarrow [Q'(<y),Q'(\leq y)]} [\varphi'(X')] }\right) \cdot Q'(y) \mathop{\mathbb{E}}_{X' \leftarrow [Q'(<y),Q'(\leq y)]} [\varphi'(X')]  \\
& \leq \sum_{y \in \mathcal{Y}} 
\mathop{\mathbb{E}}_{X \leftarrow [Q(<y),Q(\leq y)], \,\, X' \leftarrow [Q'(<y),Q'(\leq y)]}
f \left( \frac{ Q(y)  \varphi'(X)   }{ Q'(y) \varphi'(X') }\right) \cdot Q'(y) \varphi'(X') \\
& \leq \sum_{y \in \mathcal{Y}} \max_{y'}
\mathop{\mathbb{E}}_{X \leftarrow [Q(<y'),Q(\leq y')], \,\, X' \leftarrow [Q'(<y'),Q'(\leq y')]}
f \left( \frac{ Q(y)  \varphi'(X)   }{ Q'(y) \varphi'(X') }\right) \cdot Q'(y) \varphi'(X'), \\
\end{aligned}
\end{equation}
where in the first inequality we use Lemma~\ref{lem:f_convex} and Jensen's inequality. Notice that in the second inequality the maximum is taken 
only over the arguments in the expectation over $X$ and $X'$.

Jensen's inequality applies in case $X$ and $X'$ are arbitrarily coupled and
we use the same coupling between $X$ and $X'$ as in~\citep[Lemma 7,][]{papernot2022hyperparameter}, i.e., we couple $X$ and $X'$ such that
\begin{equation} \label{eq:coupling}
\frac{X - Q(< y)}{Q(y)} = \frac{X' - Q'(< y)}{Q'(y)}.
\end{equation}
We see that for $X \leftarrow [Q(<y'),Q(\leq y')]$ and $X' \leftarrow [Q'(<y'),Q'(\leq y')]$, the expressions~\eqref{eq:coupling} are between 0 and 1.
Thus, continuing from~\eqref{eq:H_alpha}, we find that
\begin{equation*} 
\begin{aligned}
H_f (A || A') & \leq \sum_{y \in \mathcal{Y}} \max_{y'}
\mathop{\mathbb{E}}_{X \leftarrow [Q(<y'),Q(\leq y')], \,\, X' \leftarrow [Q'(<y'),Q'(\leq y')]}
f \left( \frac{ Q(y)  \varphi'(X)   }{ Q'(y) \varphi'(X') }\right) \cdot Q'(y) \varphi'(X') \\
& \leq \sum_{y \in \mathcal{Y}} \max_{y'} \,
\max_{(X, X')}
f \left( \frac{ Q(y)  \varphi'(X)   }{ Q'(y) \varphi'(X') }\right) \cdot Q'(y) \varphi'(X') \\
& = \sum_{y \in \mathcal{Y}} \max_{y'}
f \left( \frac{ Q(y)  \varphi'(Q(< y') + t_y \cdot Q(y'))   }{ Q'(y) \varphi'(Q'(< y') + t_y \cdot Q'(y')) }\right) \cdot Q'(y) \varphi'(Q'(< y') + t_y \cdot Q'(y'))
\end{aligned}
\end{equation*}
for some $\{ t_y \}_{y \in \mathcal{Y}}$, where $t_y \in [0,1]$ for all $y \in \mathcal{Y}$. Furthermore, taking the maximum over $t_y$'s, we get
\begin{equation*} 
\begin{aligned}
H_f (A || A') 
& \leq \sum_{y \in \mathcal{Y}} \max_{y',t_y}
f \left( \frac{ Q(y)  \varphi'(Q(< y') + t_y \cdot Q(y'))   }{ Q'(y) \varphi'(Q'(< y') + t_y \cdot Q'(y')) }\right) \cdot Q'(y) \varphi'(Q'(< y') + t_y \cdot Q'(y')) \\
& \leq \sum_{y \in \mathcal{Y}} 
f \left( \frac{ Q(y)  \varphi'(q_y)   }{ Q'(y) \varphi'(q'_y) }\right) \cdot Q'(y) \varphi'(q'_y),
\end{aligned}
\end{equation*}
where for each $y \in \mathcal{Y}$, $q_y$ and $q'_y$ are obtained by applying the same $y$-dependent post-processing function to $Q$ and $Q'$, respectively. This can be seen using a similar reasoning as in~\citep[Lemma 7,][]{papernot2022hyperparameter}. It follows from the fact that for all $y \in \mathcal{Y}$,
there clearly exist $y_*,t_*$ such that
$$
(y_*,t_*) = \mathop{\mathrm{argmax}}_{y',t_y} f \left( \frac{ Q(y)  \varphi'(Q(< y') + t_y \cdot Q(y'))   }{ Q'(y) \varphi'(Q'(< y') + t_y \cdot Q'(y')) }\right) \cdot Q'(y) \varphi'(Q'(< y') + t_y \cdot Q'(y')).
$$
The kernel of the post-processing function is then given by
$$
g(z) = \begin{cases} 1, \quad &\textrm{if } z < y_*, \\
t_*, \quad &\textrm{if } z = y_*, \\
0, \quad &\textrm{else,} \\
\end{cases}
$$
i.e., $q_y = \sum_{z \in \mathcal{Y}} g(z) Q(z)$ and  $q'_y = \sum_{z \in \mathcal{Y}} g(z) Q'(z)$.
\end{proof}
\end{theorem}

\subsection{The Case of Continuous Output Score Function}

In case the ordered output space $\mathcal{Y}$ of the base mechanism is continuous, the proof simplifies considerably.

\begin{theorem}
For continuous output space $\mathcal{Y}$,
\begin{equation} \label{eq:thm_statement_continuous}
\begin{aligned}
H_f (A || A') 
 = \int_{\mathcal{Y}} 
f \left( \frac{ Q(y) \cdot \varphi'\big( Q( \leq y) \big)   }{ Q'(y) \cdot \varphi'\big( Q'( \leq y) \big) }\right) \cdot Q'(y) \cdot \varphi'\big( Q'( \leq y) \big) \, \dd y,
\end{aligned}
\end{equation} 
where where $Q(y)$ and $Q( \leq y)$ denote the density function the CDF of the base mechanism $Q$, respectively.
\begin{proof}
The CDF of the private selection algorithm at $y \in \mathcal{Y}$ is given by
$$
A( \leq y) = \sum_{k=1}^{\infty} \mathbf{P}[K=k] \cdot  Q( \leq y)^k,
$$
Therefore, differentiating, we see that the density function of $A$ is given by
\begin{equation} \label{eq:thm_alternative0} 
\begin{aligned}
 A(y) &= \sum_{k=1}^{\infty} k \cdot \mathbf{P}[K=k] \cdot Q(y) \cdot Q( \leq y)^{k-1} \\
&= Q(y) \sum_{k=1}^{\infty} k \cdot \mathbf{P}[K=k] \cdot Q( \leq y)^{k-1} \\
&= Q(y) \cdot \varphi'\big( Q( \leq y) \big).
\end{aligned}
\end{equation}
We have a similar representation for the density $A'(y)$, and the claim follows.
\end{proof}
\end{theorem}

\section{Proofs for Section~\ref{sec:negbin}}

\subsection{Proof of Theorem~\ref{thm:negbin}} \label{sec:appendix_negbin}

\begin{theorem}
Let $K \sim \mathcal{D}_{\eta,\gamma}$ and let $\delta(\epsilon_1)$, ${\epsilon_1} \in \mathbb{R}$, define the privacy profile of the base mechanism $Q$.
Then, for $A$ and $A'$, the output distributions of the selection algorithm evaluated on neighboring datasets $X$ and $X'$, respectively, and
for all $\epsilon_1 \geq 0$,
$$
H_{\ee^\epsilon}(A || A') \leq m \cdot \delta(\wh\epsilon)
$$
where 
\begin{equation*} 
 \wh\epsilon = \epsilon - (\eta+1)  \log \left( \ee^{\epsilon_1} + \frac{1-\gamma}{\gamma} \cdot {\delta}(\epsilon_1) \right).   
\end{equation*}

\begin{proof}
Let $q$ and $q'$ be results of applying some post-processing function to $Q$ and $Q'$, respectively.
Then,
\begin{equation} \label{eq:1_q}
1-q' \leq \ee^{\epsilon_1}(1-q) + {\delta}(\epsilon_1)
\end{equation}
for all $\epsilon_1 \in \mathbb{R}$. Thus, we see that for all $\epsilon_1 \geq 0$,
\begin{equation} \label{eq:geom_bounding1}
\begin{aligned}
    \frac{\varphi'(q)}{\varphi'(q')} &=  \left(  \frac{(\gamma-1)q' + 1}{(\gamma-1)q + 1} \right)^{\eta+1}   \\
    &=  \left( \frac{(1-\gamma)(1-q') + \gamma}{(1-\gamma)(1-q) + \gamma}   \right)^{\eta+1}  \\
    &\leq \left(  \frac{(1-\gamma)(1-q) \ee^{\epsilon_1} + \delta(\epsilon_1)(1-\gamma) + \gamma}{(1-\gamma)(1-q) + \gamma}   \right)^{\eta+1}  \\
    &\leq \left(  \frac{(1-\gamma)(1-q) \, \ee^{\epsilon_1} + \delta(\epsilon_1)(1-\gamma) + \gamma \, \ee^{\epsilon_1}}{(1-\gamma)(1-q) + \gamma}  \right)^{\eta+1}  \\
    &= \left( \ee^{\epsilon_1} + \frac{\delta(\epsilon_1)(1-\gamma)}{(1-\gamma)(1-q) + \gamma}  \right)^{\eta+1}  \\
    &\leq \left( \ee^{\epsilon_1} + \frac{1-\gamma}{\gamma} \delta(\epsilon_1) \right)^{\eta+1},
\end{aligned}    
\end{equation}
where we have used the inequality~\eqref{eq:1_q} in the first inequality and the inequality $\gamma \leq \ee^{\epsilon_1} \gamma$ in the second inequality.
Using Thm.~\ref{thm:main_thm} for the hockey-stick divergence $f(z) = [z - \ee^\epsilon]_+$, we have that for all $\epsilon_1 \geq 0$,
\begin{equation} \label{eq:geom_bounding2}
\begin{aligned}
H_f (A || A') & \leq \sum_{y \in \mathcal{Y}} 
\left[\frac{ Q(y)  \varphi'(q_y)   }{ Q'(y) \varphi'(q'_y) } - \ee^\epsilon \right]_+ \cdot Q'(y) \cdot \varphi'(q'_y) \\
& = \sum_{y \in \mathcal{Y}} 
\left[1 - \ee^{\epsilon - \log \frac{ Q(y) }{ Q'(y) } - \log \frac{ \varphi'(q_y)   }{  \varphi'(q'_y) }}\right]_+ \cdot Q(y) \cdot \frac{\mathbb{E}[K] \cdot \gamma^{\eta+1}}{\big(1 - q_y (1-\gamma)\big)^{\eta+1}}\\
& \leq \sum_{y \in \mathcal{Y}} 
\left[1 - \ee^{\epsilon - \log \frac{ Q(y) }{ Q'(y) } - \log \frac{ \varphi'(q_y)   }{  \varphi'(q'_y) }}\right]_+ \cdot Q(y) \cdot \mathbb{E}[K] \\
& \leq \sum_{y \in \mathcal{Y}} 
\left[1 - \ee^{\epsilon - \log \frac{ Q(y) }{ Q'(y) } - (\eta + 1)  \log \left( \ee^{\epsilon_1} + \frac{1-\gamma}{\gamma} \delta(\epsilon_1)    \right) }\right]_+ \cdot Q(y) \cdot \mathbb{E}[K] \\
&= \mathbb{E}[K] \cdot H_{\ee^{\wt\epsilon}} \big( Q || Q'    \big),
\end{aligned}
\end{equation}
where $\wt\epsilon = \epsilon - (\eta + 1)  \log \left( \ee^{\epsilon_1} + \frac{1-\gamma}{\gamma} \delta(\epsilon_1) \right)$.
In the third inequality we have used the inequality~\eqref{eq:geom_bounding1} the fact that $[1-\ee^{\epsilon-s}]_+$ is a non-decreasing function of $s$ for all $\epsilon \in \mathbb{R}$.
\end{proof}
\end{theorem}

\subsection{Proof of Corollaries~\ref{cor:pointwise0} and ~\ref{cor:pointwise}} 

\begin{corollary}
Let $K \sim \mathcal{D}_{\eta,\gamma}$. If the base mechanism $Q$ is $\epsilon$-DP, then the selection algorithm $A$ is $ (\eta+2) \epsilon$-DP. For $\eta=1$ we get Theorem 1.3 of~\citep{liu2019private}.
 \begin{proof}
      Let $\epsilon_1$ be such that $\delta(\epsilon_1)=0$, where
    $\big(\epsilon_1,\delta(\epsilon_1)\big)$ gives a privacy profile for the base mechanism $Q$. Then,
    \begin{equation*}
        \begin{aligned}
                \widehat{\epsilon} &= \epsilon - (\eta+1)  \log \left( \ee^{\epsilon_1} + \frac{1-\gamma}{\gamma} \delta(\epsilon_1)\right) \\
                &= \epsilon -  (\eta+1) \epsilon_1,
        \end{aligned}
    \end{equation*}
    and by Theorem~\ref{thm:negbin}, $H_{\ee^\epsilon}(A || A')=0$ if
    $\epsilon =  (\eta+2)\epsilon_1$.
 \end{proof}
\end{corollary}

\begin{corollary} 
 Let $K \sim \mathcal{D}_{\eta,\gamma}$. If the base mechanism $Q$ is $(\epsilon,\delta)$-DP, then then the selection algorithm $A$ is $ \big((\eta+2) \epsilon + \gamma^{-1}  \delta, m \delta \big)$-DP.
 \begin{proof}
 Let $\big(\epsilon_1,\delta(\epsilon_1)\big)$ be a privacy profile for the base mechanism $Q$
and $\epsilon_1 \geq 0$. Then,
    \begin{equation*}
        \begin{aligned}
                \widehat{\epsilon} &= \epsilon - (\eta+1)  \log \left( \ee^{\epsilon_1} + \frac{1-\gamma}{\gamma} \delta(\epsilon_1)\right) \\
                &\geq \epsilon - (\eta+1)  \log \left( \ee^{\epsilon_1}\left(1 + \gamma^{-1} \delta(\epsilon_1)\right) \right) \\
                &\geq \epsilon - (\eta+1) \left( \epsilon_1 + \gamma^{-1} \delta(\epsilon_1)\right) \\
        \end{aligned}
    \end{equation*}
    and by Thm.~\ref{thm:negbin}, $H_{\ee^\epsilon}(A || A')= m \delta$ if
    $\epsilon =  (\eta+2)\epsilon_1 +  \gamma^{-1}  \delta$.
 \end{proof}
\end{corollary}

\subsection{Proof of Corollary~\ref{lem:GDP}} 

\begin{lemma}[$\epsilon$-values when $Q$ is GDP]

Let $K \sim \mathcal{D}_{\eta,\gamma}$. 
Suppose the base mechanism is dominated by the Gaussian mechanism with noise parameter $\sigma > 0$ and $L_2$-sensitivity 1. Then, for a fixed $\delta>0$, the private selection algorithm $A$ is $(\epsilon,\delta)$-DP for
$$
\epsilon = (\eta+2) \left( \frac{1}{2 \sigma^2} + \frac{1}{\sigma} \sqrt{ 2 \log \frac{1}{\gamma \cdot \delta}}\right) + \delta.
$$
\begin{proof}
By~\citep[Lemma 3,][]{balle2018improving}, we know that for the Gaussian mechanism with sensitivity 1 and noise scale $\sigma$ the privacy loss random variable $\omega$ is distributed as
$\omega \sim \mathcal{N}\left( \frac{1}{2 \sigma^2}, \frac{1}{\sigma^2} \right).$
Thus, using a simple Chernoff bound for the Gaussian, for its privacy profile $\delta(\epsilon_1)$ we have
$$
\delta(\epsilon_1) \leq \mathbb{P}(\omega \geq \epsilon_1) \leq \ee^{ - \frac{\wt\epsilon^2 \sigma^2}{2}},
$$
where $\wt \epsilon = \epsilon_1 -  \frac{1}{2 \sigma^2}$.
Choosing
$$
\epsilon_1 = \frac{1}{2 \sigma^2} + \frac{1}{\sigma} \sqrt{ 2 \log \frac{m}{\delta} },
$$
we see that 
$$
\delta(\epsilon_1) \leq \frac{\delta}{m}.
$$
Furthermore, since $\epsilon_1 \geq 0$, we have the following bound for the additional term in the bound of Thm.~\ref{thm:negbin}:
\begin{equation} \label{eq:GDP_proof}
    \begin{aligned}
    (\eta+1)  \log \left( \ee^{\epsilon_1} + \tfrac{1-\gamma}{\gamma} \delta(\epsilon_1)\right) & \leq 
      (\eta+1)  \log \left( \ee^{\epsilon_1} + \ee^{\epsilon_1} \tfrac{1-\gamma}{\gamma} \delta(\epsilon_1)\right) \\
     & =   (\eta+1) \epsilon_1 + \log\big(1 + \tfrac{1-\gamma}{\gamma} \delta(\epsilon_1) \big) \\
        & \leq 
        (\eta+1) \epsilon_1 + \log\big(1 + \gamma^{-1} \cdot \delta(\epsilon_1) \big) \\
        & \leq 
        (\eta+1) \epsilon_1 + \log\big(1 + \delta \big) \\
        & \leq (\eta+1) \epsilon_1 + \delta.
    \end{aligned}
\end{equation}
Setting $\epsilon = (\eta+2) \epsilon_1 + \delta$, the inequality \eqref{eq:GDP_proof} and  Thm.~\ref{thm:negbin} show that
\begin{equation*}
    \begin{aligned}
        H_{\ee^\epsilon}(A || A') 
    \leq & m \cdot \delta\big(\epsilon -(\eta+1) \epsilon_1 + \log\big(1 + \tfrac{1-\gamma}{\gamma} \delta(\epsilon_1) \big)\big) \\
    \leq & m \cdot \delta(\epsilon_1) \leq \delta.
    \end{aligned}
\end{equation*}
\end{proof}
\end{lemma}


\subsection{Proof of Theorem~\ref{thm:hs_bin} (Private Selection, Binomial Distribution)} 

First, the following auxiliary lemma is needed.

\begin{lemma} \label{lem:aux0}
    Let $a,b,c,d > 0$. Then, for $x \geq 0$, the function
    $$
    f(x) = \frac{ax + b}{cx + d}
    $$
    is non-decreasing if and only if $\tfrac{a}{b} \geq \tfrac{c}{d}$.
    \begin{proof}
        The claim follows from the expression
        $$
        f'(x) = \frac{ad - cb}{(cx + d)^2}.
        $$
    \end{proof}
\end{lemma}

Recall that for $K \sim \mathrm{Bin}(n,p)$, the probability generating function is given by 
\begin{equation} \label{eq:PGF_binomial}
    \varphi(z) = (1 - p + pz)^n.
\end{equation}

\begin{theorem} 
Let $K \sim \mathrm{Bin}(n,p)$ for some $n \in \mathbb{N}$ and $0< p < 1$, and let $\delta({\epsilon_1})$, ${\epsilon_1} \in \mathbb{R}$, define the privacy profile of the base mechanism $Q$.
Suppose 
$$
{\epsilon_1} \geq \log \left( 1 + \tfrac{p}{1-p} \delta(\epsilon_1) \right).
$$
Then, for $A$ and $A'$, the output distributions of the selection algorithm evaluated on neighboring datasets $X$ and $X'$, respectively, for all $\epsilon > 0$ and
for all $\epsilon_1 \geq 0$,
\begin{equation} \label{eq:bin_bound_0A}
    H_{\ee^\epsilon}(A || A') \leq m \cdot \delta(\wh{\epsilon}), 
\end{equation}
where 
$$
\wh{\epsilon} = \epsilon - (n-1) \log \big(  1 +  p (\ee^{\epsilon_1} - 1) + p \delta(\epsilon_1)  \big).
$$

\begin{proof}

Using the PGF of the binomial distribution given in Eq.~\eqref{eq:PGF_binomial}
by the auxiliary Lemma~\ref{lem:aux0}, for each $y \in \mathcal{Y}$, 
\begin{equation} \label{eq:bin1}
\begin{aligned}
\frac{\varphi'(q_y)}{ \varphi'(q'_y) } &=  \left( \frac{ 1 - p + p q_y}{1 - p + p q_y'} \right)^{n-1} \\
& \leq \left( \frac{ 1 - p + p \ee^{\epsilon_1} q_y' + p \delta(\epsilon_1) }{1 - p + p q_y'} \right)^{n-1} \\
& = \left( 1 + \frac{  p (\ee^{\epsilon_1} - 1) q_y' + p \delta(\epsilon_1) }{1 - p + p q_y'} \right)^{n-1} \\
& \leq \big( 1 +  p (\ee^{\epsilon_1} - 1) + p \delta(\epsilon_1) \big)^{n-1}, \\
\end{aligned}
\end{equation}
in case
$$
\frac{p (\ee^{\epsilon_1} - 1)}{p \delta(\epsilon_1)} \geq \frac{p}{1-p},
$$
i.e., if 
$$
{\epsilon_1} \geq \log \left( 1 + \frac{p}{1-p} \delta(\epsilon_1) \big)  \right).
$$
Moreover, we see that
\begin{equation} \label{eq:bin2}
\varphi'(q) \leq n \cdot p = m
\end{equation}
for all $0 \leq q \leq 1$. The claim follows from the inequalities~\eqref{eq:bin1} and~\eqref{eq:bin2} and from Theorem~\ref{thm:main_thm}. 
\end{proof}
\end{theorem}

\subsection{Proof of Corollary~\ref{cor:binomial_to_poisson}}
\begin{corollary}
Let $K \sim \mathrm{Poisson}(m)$ for some $m \in \mathbb{N}$, and let $\delta({\epsilon_1})$, ${\epsilon_1} \in \mathbb{R}$, define the privacy profile of the base mechanism $Q$.
Then, for $A$ and $A'$, the output distributions of the selection algorithm evaluated on neighboring datasets $X$ and $X'$, respectively, and
for all $\epsilon > 0$, and for all $\epsilon_1 \geq 0$,
\begin{equation*} 
    H_{\ee^\epsilon}(A || A') \leq m \cdot \delta(\wh{\epsilon}), 
\end{equation*}
where 
$$
\wh{\epsilon} = \epsilon - m \cdot ( \ee^{\epsilon_1} - 1 ) - m \cdot \delta(\epsilon_1).
$$    
\begin{proof}
Let $K \sim \mathrm{Poisson}(m)$ and $K_n \sim \mathrm{Bin}(n,m/n)$ and let $A,A'$ denote the density functions of the private selection algorithm corresponding to $K$ and let $A_n,A_n'$ those corresponding to $K_n$, evaluated on neighboring datasets $X,X'$, respectively. Looking at the form of $A$ given in Eq.~\eqref{eq:RDP_A0}, we have that
\begin{equation*}
    \begin{aligned}
        H_{\ee^\epsilon}\big( A || A' \big) &= \sum_{y \in \mathcal{Y}} \max\{ A(y) - \ee^\epsilon A'(y),0 \} \\
        &=\sum_{k=1}^{\infty} \max\{ \mathbf{P}[K=k]  \cdot \big( Q( \leq y)^k - Q(<y)^k \big) - \ee^\epsilon \mathbf{P}[K=k] \cdot \big( Q'( \leq y)^k - Q'(<y)^k \big) , 0 \} \\
        &\leq \sum_{k=1}^{\infty} \max\{ \mathbf{P}[K_n=k] \cdot \big( Q( \leq y)^k - Q(<y)^k \big) - \ee^\epsilon \mathbf{P}[K_n=k] \cdot \big( Q'( \leq y)^k - Q'(<y)^k \big), 0 \}  \\
        & \quad + (1+\ee^\epsilon) \sum_{k=0}^\infty \abs{\mathbf{P}[K=k] - \mathbf{P}[K_n=k]} \\
        & = H_{\ee^\epsilon}\big( A_n || A_n' \big) + (1+\ee^\epsilon) \sum_{k=0}^\infty \abs{\mathbf{P}[K=k] - \mathbf{P}[K_n=k]},
    \end{aligned}
\end{equation*}
where the inequality follows from the fact that $\max\{ a + b,0 \} \leq \abs{a} + \max\{ b ,0 \} $ for all $a,b \in \mathbb{R}$.
Since by Le Cam's inequality,  $\mathrm{Bin}(n,m/n) \rightarrow \mathrm{Poisson}(m)$ in total variation distance as $n \rightarrow \infty$, we have that
$$
H_{\ee^\epsilon}\big( A || A' \big) = \lim_{n \rightarrow \infty} H_{\ee^\epsilon}\big( A_n || A_n' \big).
$$
Fixing $n \cdot p = m$ in the bound~\eqref{eq:bin_bound_0} of Thm.~\ref{thm:hs_bin} (bound for the case $K \sim \mathrm{Bin}(n,n/m)$), we see that the bound approaches the bound \eqref{eq:poisson_bound_0} of Cor.~\ref{cor:binomial_to_poisson} (bound for the case $K \sim \mathrm{Poisson}(m)$) as $p \rightarrow 0$, since then
$$
(n-1) \log \big(  1 +  p (\ee^{\epsilon_1} - 1) + p \delta(\epsilon_1)  \big)
\rightarrow m \cdot (\ee^{\epsilon_1} - 1) + m \cdot \delta(\epsilon_1).
$$
This follows from the fact that $\frac{\log (1+x)}{x} \rightarrow 1$ as $x \rightarrow 0$.
\end{proof}
\end{corollary}

\begin{remark}
    We can also get Cor.~\ref{cor:binomial_to_poisson} directly using the PGF of the Poisson distribution.
    For $K \sim \mathrm{Poisson}(m)$, the PGF is $\varphi(z) = \ee^{m (z-1)}$, i.e. $\varphi'(z) = m \cdot \ee^{m (z-1)}$.
Using Thm.~\ref{thm:main_thm} for the hockey-stick divergence $f(z) = [a - \ee^\epsilon]_+$, we get
\begin{equation} \label{eq:HS_bounding2}
\begin{aligned}
H_f (A || A') & \leq \sum_{y \in \mathcal{Y}} 
f \left( \frac{ Q(y)  \varphi'(q_y)   }{ Q'(y) \varphi'(q'_y) }\right) \cdot Q'(y) \cdot \varphi'(q'_y) \\
& = \sum_{y \in \mathcal{Y}} 
\left[\frac{ Q(y)  \varphi'(q_y)   }{ Q'(y) \varphi'(q'_y) } - \ee^\epsilon\right]_+ \cdot Q'(y) \cdot \varphi'(q'_y) \\
& = \sum_{y \in \mathcal{Y}} 
\left[1 - \ee^{\epsilon - \log \frac{ Q(y) }{ Q'(y) } - \log \frac{ \varphi'(q_y)   }{  \varphi'(q'_y) }}\right]_+ \cdot Q(y) \cdot \varphi'(q_y) \\
& = \sum_{y \in \mathcal{Y}} 
\left[1 - \ee^{\epsilon - \log \frac{ Q(y) }{ Q'(y) } - m \cdot ( q_y- q'_y) }\right]_+ \cdot Q(y) \cdot m \cdot  \ee^{m \cdot (q_y-1) }. \\
\end{aligned}
\end{equation}
As the probabilities $q_y$ and $q'_y$ are obtained by applying the same post-processing to $Q$ and $Q'$,
for all ${\epsilon_1} \geq 0$,
$q_y' \leq \ee^{\epsilon_1} q_y + \delta(\epsilon_1)$.
Using also the fact that $[1-\ee^{\epsilon-s}]_+$ is a non-decreasing function of $s$ for all $\epsilon \in \mathbb{R}$, we get from the inequality~\eqref{eq:HS_bounding2} that for all $\epsilon_1 \geq 0$,
\begin{equation} \label{eq:HS_bounding4}
\begin{aligned}
H_f(A || A') & \leq  \sum_{y \in \mathcal{Y}} 
\left[1 - \ee^{\epsilon - \log \frac{ Q(y) }{ Q'(y) } - m \cdot ( q_y- q'_y) }\right]_+ \cdot Q(y) \cdot m \cdot  \ee^{m \cdot (q_y-1) } \\
& \leq  \sum_{y \in \mathcal{Y}} 
\left[1 - \ee^{\epsilon - \log \frac{ Q(y) }{ Q'(y) } - m \cdot  q'_y (1 - \ee^{\epsilon_1}) - m \cdot \delta(\epsilon_1)}\right]_+ \cdot Q(y) \cdot m \cdot  \ee^{m \cdot (q_y-1) }. \\
& \leq  \sum_{y \in \mathcal{Y}} 
\left[1 - \ee^{\epsilon - \log \frac{ Q(y) }{ Q'(y) } - m \cdot  (1 - \ee^{\epsilon_1}) - m \cdot \delta(\epsilon_1)}\right]_+ \cdot Q(y) \cdot m
\end{aligned}
\end{equation}
\end{remark}
which gives the claim of Cor.~\ref{cor:binomial_to_poisson}.

\section{Converting RDP Bounds to $(\epsilon,\delta)$-Bounds}

To convert from R\'enyi DP to approximate DP we use following formula.

\begin{lemma}[\citealt{canonne2020discrete}]  \label{lem:rdp_to_dp}
Suppose the mechanism $\mathcal{M}$ is $\big(\alpha,\epsilon' \big)$-RDP.
Then  $\mathcal{M}$ is also $(\epsilon,\delta(\epsilon))$-DP for arbitrary $\epsilon\geq 0$ with
\begin{equation} \label{eq:conversion_canonne}
\delta(\epsilon) =    \frac{\exp\big( (\alpha-1)(\epsilon' - \epsilon) \big)}{\alpha} \left( 1 - \frac{1}{\alpha}  \right)^{\alpha-1}.
\end{equation}
\end{lemma}

\end{document}